\pdfoutput=1

\documentclass[11pt]{article}

\usepackage[final]{acl}

\usepackage{times}
\usepackage{latexsym}
\usepackage{authblk}
\usepackage{adjustbox}
\usepackage{booktabs}
\usepackage[T1]{fontenc}

\usepackage[utf8]{inputenc}

\usepackage{microtype}

\usepackage{inconsolata}

\usepackage{graphicx}

\title{On the Vulnerability of Text Sanitization}

\usepackage{colortbl}
\usepackage{multirow}
\definecolor{light_red}{rgb}{0.96, 0.76, 0.76}
\definecolor{light_green}{rgb}{0.84, 0.94, 0.77}
\usepackage{subfig}
\usepackage{xspace}
\usepackage{amsmath} 
\usepackage[english]{babel}
\usepackage{amsthm}
\usepackage{cleveref}
\newcommand{\mypara}[1]{\smallskip\noindent\textbf{#1.} \xspace}

\newtheorem{theorem}{\bf Theorem}
\newtheorem{theorema}{\bf Theorem}
\newtheorem{definition}{\bf{Definition}}
  
\usepackage{soul}

\author[1]{Meng Tong}
\author[1\thanks{Corresponding author.}]{Kejiang Chen}
\author[1]{Xiaojian Yuan}
\author[2]{Jiayang Liu}
\author[1]{\\Weiming Zhang}
\author[1]{Nenghai Yu}
\author[3]{Jie Zhang}

\affil[1]{University of Science and Technology of China}
\affil[2]{Nanyang Technological University}
\affil[3]{CFAR and IHPC, A*STAR, Singapore}
\affil[ ]{\{\texttt{mtong,xjyuan\}@mail.ustc.edu.cn,\{\texttt{chenkj,zhangwm,ynh\}@ustc.edu.cn}}}
\affil[ ]{\texttt{jiayang.liu@ntu.edu.sg}, \; \texttt{zhang\_jie@cfar.a-star.edu.sg}}

\begin{document}
\maketitle
\begin{abstract}
Text sanitization, which employs differential privacy to replace sensitive tokens with new ones, represents a significant technique for privacy protection. Typically, its performance in preserving privacy is evaluated by measuring the attack success rate (ASR) of reconstruction attacks, where attackers attempt to recover the original tokens from the sanitized ones. However, current reconstruction attacks on text sanitization are developed empirically, making it challenging to accurately assess the effectiveness of sanitization. In this paper, we aim to provide a more accurate evaluation of sanitization effectiveness. Inspired by the works of Palamidessi\;et\;al.~\cite{alvim2015information, cherubin2019f}, we implement theoretically optimal reconstruction attacks targeting text sanitization. We derive their bounds on ASR as benchmarks for evaluating sanitization performance. For real-world applications, we propose two practical reconstruction attacks based on these theoretical findings. Our experimental results underscore the necessity of reassessing these overlooked risks. Notably, one of our attacks achieves a $46.4\%$ improvement in ASR over the state-of-the-art baseline, with a privacy budget of $\epsilon=4.0$ on the SST-2 dataset. Our code is available at: \href{https://github.com/mengtong0110/On-the-Vulnerability-of-Text-Sanitization}{https://github.com/mengtong0110/On-the-Vulnerability-of-Text-Sanitization}.
\end{abstract}

\section{Introduction}\label{intr}

With the advancement of natural language processing (NLP)~\cite{chang2023survey}, concerns related to its privacy implications~\cite{kim2024propile} have significantly intensified. For instance, Samsung employees previously leaked the company’s confidential meeting records by uploading documents to the cloud services of large language models (LLMs)~\cite{Sam}. More recently, reports~\cite{report,report2} have disclosed a concerning flaw in LLM-based chatbots~\cite{gpt_store}, which could expose user data and put private information at risk. 
\begin{figure}[t]
\centering
\includegraphics[width=1\columnwidth]{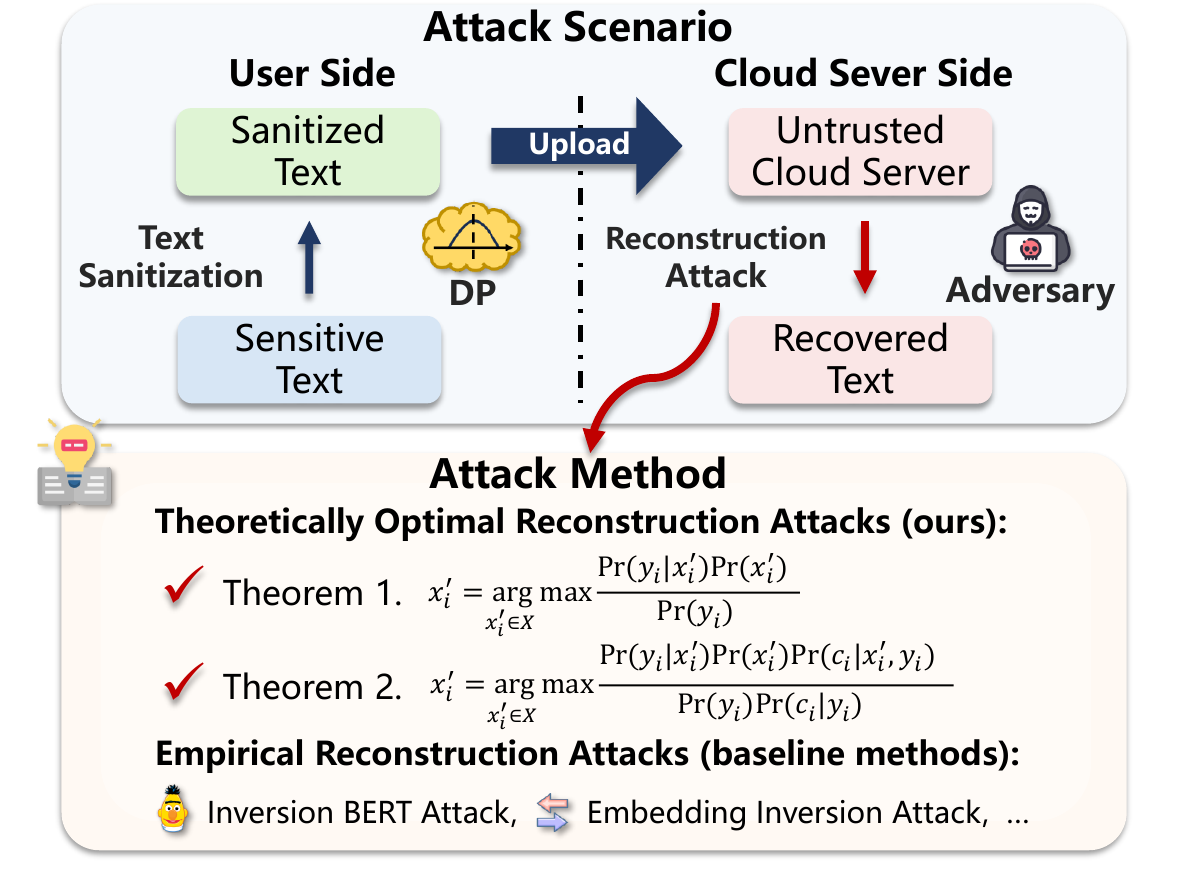}
\vspace{-0.4cm}
\caption{The illustration of reconstruction attacks.}
\label{fig:problem}
\end{figure}

To\,address\,these\,privacy\,concerns, text sanitization \cite{chen2023customized,tong2023InferDPT} is now regarded as a crucial solution for privacy preservation in text and has been widely adopted by companies, such as Amazon \cite{amazon} and Microsoft \cite{Priva}. Specifically, before users upload their documents to an untrusted cloud server for model training or inference, they can employ differential privacy (DP) \cite{dwork2006differential} to replace tokens corresponding to sensitive entities in the documents with newly sampled tokens. This method minimizes alterations in the uploaded data compared to other methods, such as synthetic data generation \cite{hong2024dpopt,utpala-etal-2023-locally}.

While DP sanitizes sensitive tokens, recent works~\cite{zhou2023textobfuscator,du2023dp} have demonstrated that attackers can still recover them from sanitized text. This finding underscores the significance of measuring the effectiveness of text sanitization. One of the principal methods for evaluating its effectiveness in privacy protection is through reconstruction attacks~\cite{balle2022reconstructing,yue-etal-2021-differential}, which aim to recover sensitive tokens from sanitized text, as illustrated in \Cref{fig:problem}.

\mypara{Existing Attacks} In this context, several reconstruction attacks have been developed to evaluate the vulnerability of text sanitization: The \textit{Embedding Inversion Attack} \cite{qu2021natural} reconstructs tokens by selecting the token with the closest embedding distance to the sanitized token \cite{song2020information}. The \textit{Inversion BERT Attack} \cite{kugler2021invbert} trains an inversion model to map sanitized tokens back to their original forms. The \textit{Mask Token Inference Attack} \cite{yue-etal-2021-differential} replaces each token in the sanitized text with a “[MASK]” token to predict the raw one. However, these attacks are empirically developed, which makes it difficult to accurately measure the effectiveness of sanitization (as discussed in \Cref{rw}).

\mypara{Our Proposal} {In this paper, we aim to address the evaluation gap in privacy protection for DP-based text sanitization.} Inspired by the works of Palamidessi et al.~\cite{alvim2015information, cherubin2019f}, who estimate privacy leakage using Bayesian inference, we introduce both the \textit{Context-free Optimal Reconstruction Attack} and the \textit{Contextual Optimal Reconstruction Attack} for assessing text sanitization. We also derive theoretical bounds on the attack success rate (ASR), which serve as benchmarks for evaluating its performance in protecting privacy. To support real-world implementation, we propose two practical attacks based on these theoretical findings. Experimental results demonstrate that our reconstruction attacks achieve higher ASRs than baseline methods.

\mypara{Our Contributions} We summarize our main contributions as follows:
\vspace{-0.2cm}
\begin{enumerate}
\item We implement theoretically context-free and contextual optimal reconstruction attacks against text sanitization, deriving their bounds on ASR to serve as benchmarks for evaluating the effectiveness of the sanitization.
\vspace{-0.33cm}
\item Based on these theoretical findings, we~propose two practical reconstruction attacks: \textit{Context-free~Bayesian~Attack} and \textit{Contextual Bayesian Attack}. Notably, \textit{Contextual Bayesian Attack} is the first to transform the reconstruction attack to a  classification task.
\vspace{-0.33cm}
\item We reassessed text sanitization methods across four datasets using seven reconstruction attacks. Our experimental results underscore the necessity of paying closer attention to these overlooked risks. Specifically, \textit{Contextual Bayesian Attack} achieves a 46.4\% improvement in ASR over the state-of-the-art baseline at $\epsilon=4.0$ on the SST-2 dataset~\cite{socher-etal-2013-recursive}.
\end{enumerate}

\vspace{-0.4cm}
\section{Preliminaries and Related Work}\label{rw}
\vspace{-0.1cm}
\mypara{Text Sanitization}\!\!Text sanitization~\cite{chen2023customized,yue-etal-2021-differential} leverages differential privacy (DP) \cite{dwork2006differential} to replace sensitive tokens with newly sampled tokens. Specifically, let \( X \) represent the set of sensitive tokens and \( Y \) represent the set of sanitized tokens. An original sentence consists of sensitive tokens \( x_1, x_2, \ldots, x_n \), where \(x_i \in X\). Its sanitized sentence is \( y_1, y_2, \ldots, y_n \), where the token \(y_i \in Y\) is independently sampled by DP with the probability \(\Pr\left(y_i|x_i\right)\). Additionally, we define the sanitized context \( \mathbf{c}_{i} = y_{1}, y_{2}, \ldots, y_{i-1}, y_{i+1}, \ldots, y_n \) for the sanitized token $y_{i}$ ({or original token $x_{i}$}), where \( \mathbf{c}_{i} \in C\) and \(C\) denotes the set of sanitized contexts.

In addition to text sanitization, recent studies~\cite{staab2024large, StaabVBV24} have demonstrated that LLMs can also protect private information.

\mypara{Definition of Sensitive Tokens} 
In previous research~\cite{tong2023InferDPT,papadopoulou2022neural,plant2021cape}, they have made varying choices on the definition of sensitive tokens. In our study, we set the sensitive tokens in accordance with the corresponding defense method of text sanitization~\cite{chen2023customized,yue-etal-2021-differential}.



\mypara{Reconstruction Attacks}
The \textit{Embedding Inversion Attack}~\cite{qu2021natural} identifies the token closest to a sanitized token by comparing their embedding distances~\cite{song2020information} as the reconstructed token. The \textit{Inversion BERT Attack}~\cite{kugler2021invbert} uses a trained model to convert sanitized representations back to their original tokens on a one-to-one basis. However, neither method has performed a comprehensive evaluation of sanitization effectiveness, which should consider the context of the sanitized tokens. The \textit{Mask Token Inference Attack}~\cite{yue-etal-2021-differential} replaces each token in sanitized text with a ``[MASK]'' token, which is then predicted to reconstruct the raw token. This attack overlooks the impact of probability distributions of raw tokens and the sampling probabilities in DP on the attack result, leading to an inaccurate evaluation of sanitization performance.

To address these issues, we introduce optimal reconstruction attacks and propose corresponding practical approaches for evaluating the effectiveness of DP-based text sanitization.

\section{Theoretically Optimal Reconstruction Attacks}
In this section, we first present the threat model and then develop \textit{Context-free Optimal Reconstruction Attack}, deriving its ASR bound to enable quick evaluation of text sanitization. For a comprehensive evaluation, we propose \textit{Contextual Optimal Reconstruction Attack} with its ASR bound.

\subsection{Threat Model for Theoretical Attacks} 
\mypara{Adversary's Goal} The goal of an adversary is to obtain information on sensitive tokens by reconstructing each original token from its sanitized one.   

\mypara{Adversary's Knowledge} We assume an adversary who knows the differentially private algorithm with its sampling probabilities \(\Pr\left(Y|X\right)\), which is utilized in text sanitization. Additionally, we propose an assumption that the adversary possesses the following theoretical probabilities: \(\Pr\left({X}\right)\) (the probability of a token appearing in an original sentence) and \(\Pr\left(C|X, Y\right)\) (the probability that a sanitized context corresponds to an original token and a sanitized token in the sanitization process). 

\subsection{Warm-up: Context-free Optimal Reconstruction Attack}
Building on the Bayesian estimation of DP~\cite{alvim2015information, cherubin2019f}, we introduce the context-free optimal strategy for reconstructing the original tokens from sanitized tokens. We further derive its tight ASR bound, which achieves a quick evaluation of sanitization effectiveness. 

\begin{theorem}[\textbf{Context-free Optimal Reconstruction Attack}]\label{theorem:1}
Given a sanitized token $y\in Y$, the optimal strategy to reconstruct the original token of $y$ is to select the token ${x}' \in X$ according to the following rule:
\begin{align}
{x}' =&\, \underset{{x}' \in X}{\mathrm{arg\,max\,}} \Pr\left({x}'|y\right)\label{eq1}\\ 
=& \,\underset{{x}' \in X}{\mathrm{arg\,max\,}} \frac{\Pr\left(y|{x}'\right)  \Pr\left({x}'\right)}{\Pr\left(y\right)}.
\end{align}
\end{theorem}

The proof of \Cref{theorem:1} is in~\Cref{prf}. It is shown that an adversary should select the most likely original token ${x}'$ corresponding to the sanitized token $y$ as the reconstructed token. According to Bayes' formula~\cite{joyce2003bayes}, \Cref{eq1} can be decomposed into several components. $\Pr\left(y\right)$ is constant for all ${x}'\in X$ and does not affect the outcome of the attack. For $\Pr\left(y|{x}'\right)$ and $\Pr\left({x}'\right)$, it is observed that: 1) A token ${x}'$ with a higher sampling probability $\Pr\left(y|{x}'\right)$ is more likely to be the original token of $y$; 2) A token ${x}'$ with a higher $\Pr\left({x}'\right)$ appears more frequently and is more likely to be the raw token of $y$.

Based on the context-free optimal reconstruction attack, we derive its tight bound on the ASR as a quick evaluation for the effectiveness of DP-based text sanitization:

\begin{definition}[\textbf{Context-free ASR Bound}]
Let $x_1^j$\,, $x_2^j$\,, $\ldots$\,, $x_{n_j}^j$  and $y_1^j$\,, $y_2^j$\,, \ldots\,, $y_{n_j}^j$ represent the \(j\)-th original sentence and its sanitized sentence by DP. Given any $m$ sanitized sentences, the Context-free Attack Success Rate Bound (Context-free ASR Bound) is defined as follows:
\begin{equation}
{
\hspace{-6pt}\hspace{-3pt}\text{Context-free ASR Bound}:=\frac{\sum_{j=1}^{m}\sum_{i=1}^{n_j} \delta_{x_{i}^j,{x_{i}^j}'}}{\sum_{i=1}^{m}n_j},}
\end{equation}
\begin{equation}
\hspace{-6pt}{x_{i}^j}'\hspace{-1pt}=\hspace{-1pt}\underset{{x_{i}^j}' \in X}{\mathrm{arg\,max\,}}\hspace{-2pt}\Pr\,({x_{i}^j}'|y_{i}^j),\delta_{x_{i}^j, {x_{i}^j}'} = \begin{cases} 
1\;,x_{i}^j={x_{i}^j}' \\
0\;,x_{i}^j\neq {x_{i}^j}'\end{cases}\hspace{-13pt}.
\end{equation}

 \end{definition}
 
The \textit{Context-free ASR Bound} serves as a tight ASR bound for the \textit{Context-free Optimal Reconstruction Attack}. This is because the attack of \textit{Context-free ASR Bound} strictly follows \Cref{theorem:1} to reconstruct the sanitized tokens. This bound primarily utilizes simple frequency statistics, making it a convenient tool for swiftly evaluating DP-based text sanitization methods. To represent $\Pr({x_{i}^j}')$, which is informed to the adversary in the threat model, it computes the theoretical probability of the original token ${x_{i}^j}'$ appearing in the original sentences. The term $\Pr(y_{i}^j|{x_{i}^j}')$ denotes the probability that DP samples sanitized token $y_{i}^j$ with original token ${x_{i}^j}'$. Although \textit{Context-free ASR Bound}  swiftly assesses the sanitization effectiveness, it does not consider sanitized context, which may lead to a less comprehensive evaluation. Therefore, we develop \textit{Contextual Optimal Reconstruction Attack}.

\subsection{Advanced: Contextual Optimal Reconstruction Attack}

\begin{theorem}[\textbf{Contextual Optimal Reconstruction Attack}]\label{theorem:2}
Given a sanitized token $y$ and its sanitized context $\mathbf{c}$, the optimal strategy to reconstruct the original token of $y$ is to select the token ${x}'\in X$ according to the following rule:
{\begin{align}
&{x}' = \underset{{x}' \in X}{\mathrm{arg\,max\,}}  \Pr\left({x}'|y,\mathbf{c}\right)\\&= \underset{{x}' \in X}{\mathrm{arg\,max\,}}  \frac{\Pr\left(y|{x}'\right) \Pr\left({x}'\right)\Pr\left(\mathbf{c}|{x}', y\right)}{\Pr\left(y\right)\Pr\left(\mathbf{c}|y\right)} .
\end{align}}
\end{theorem}

The proof of \Cref{theorem:2} is provided in \Cref{prf}. It demonstrates that an adversary should select the most probable original token ${x}'$, corresponding to the sanitized token $y$ and its sanitized context $\mathbf{c}$ (previously defined in \Cref{rw}), as the reconstructed token. \Cref{theorem:2} reduces the amount of uncertainty in reconstructing the original tokens compared to \Cref{theorem:1}. This holds because \(  H\!\left(X|\,Y,C\right) \leq H\!\left(X|\,Y\right) \), where \(  H\!\left(X|\,Y\right) \) is the entropy of \( X \) conditioned on \( Y \), and \(  H\!\left(X|\,Y,C\right) \) is the entropy of \( X \) conditioned on both \( Y \) and \( C \). Additionally, \Cref{theorem:2} includes extra calculations for $\Pr(\mathbf{c}|{x}', y)$ and $\Pr\left(\mathbf{c}|y\right)$ relative to \Cref{theorem:1}. $\Pr\left(\mathbf{c}|y\right)$ remains constant for all ${x}'\in X$ and does not influence the outcome of the attack. $\Pr\left(\mathbf{c}|{x}', y\right)$ represents the probability that $\mathbf{c}$ is the sanitized context of both ${x}'$ and $y$.

\mypara{Challenge of Computation}However, since $\mathbf{c}$ is a sequence that consists of sanitized tokens, accurately calculating the theoretical value of $\Pr\left(\mathbf{c}|{x}', y\right)$ poses a challenge. Due to this intractability, we use its approximation to derive the ASR bound of attack strategy in \Cref{theorem:2}.

\begin{definition}[\textbf{Contextual \textit{K}-ASR Bound}]
Let $x_1^j$\,, $x_2^j$\,, $\ldots$\,, $x_{n_j}^j$  and $y_1^j$\,, $y_2^j$\,, \ldots\,, $y_{n_j}^j$ represent the \(j\)-th original sentence and its sanitized sentence by DP. Let the sanitized context \( \mathbf{c}_{i}^j\) for the \(i\)-th sanitized token $y_{i}^j$ (whose original token is $x_{i}^j$) be $\mathbf{c}_{i}^j$ $=$ $y_1^j$, $y_2^j$, $\ldots$\,, $y_{i-1}^j$, $y_{i+1}^j$, $\ldots$\,, $y_{n_j}^j$. Given any $m$ sanitized sentences, the Contextual K Attack Success Rate Bound (Contextual \textit{K}-ASR Bound) is defined as follows:
\begin{equation}
{
\text{Contextual \textit{K}-ASR Bound} := \frac{\sum_{j=1}^{m}\sum_{i=1}^{n_j} \delta_{x_{i}^j,{x_{i}^j}'}}{\sum_{j=1}^{m}n_j},}    
\end{equation}
\begin{equation}
{x_{i}^j}' = \underset{{x_{i}^j}' \in {X_{i}^j}'}{\mathrm{arg\,max\,}} \Pr({x_{i}^j}'|y_{i}^j,\mathbf{c}_{i}^j),    
\end{equation}
\begin{equation}
\hspace{-3pt}{X_{i}^j}'\hspace{-3pt}=\!\underset{{x_{i}^j}' \in X}{\mathrm{arg\,top\,}K} \Pr\,({x_{i}^j}'|y_{i}^j),\,\delta_{x_{i}^j, {x_{i}^j}'} = \begin{cases} 
1\;,x_{i}^j={x_{i}^j}' \\
0\;,x_{i}^j\neq {x_{i}^j}'\end{cases}\hspace{-13pt}.
\end{equation}
\end{definition}
\begin{figure*}[t]
\hspace*{9pt}
    \setlength{\tabcolsep}{-6pt} 
    \centering
    \begin{tabular}{ccc}
    \subfloat[\textit{SANTEXT+}\_SST-2\hspace*{3.4pt}]{
        \includegraphics[width=0.35\textwidth]{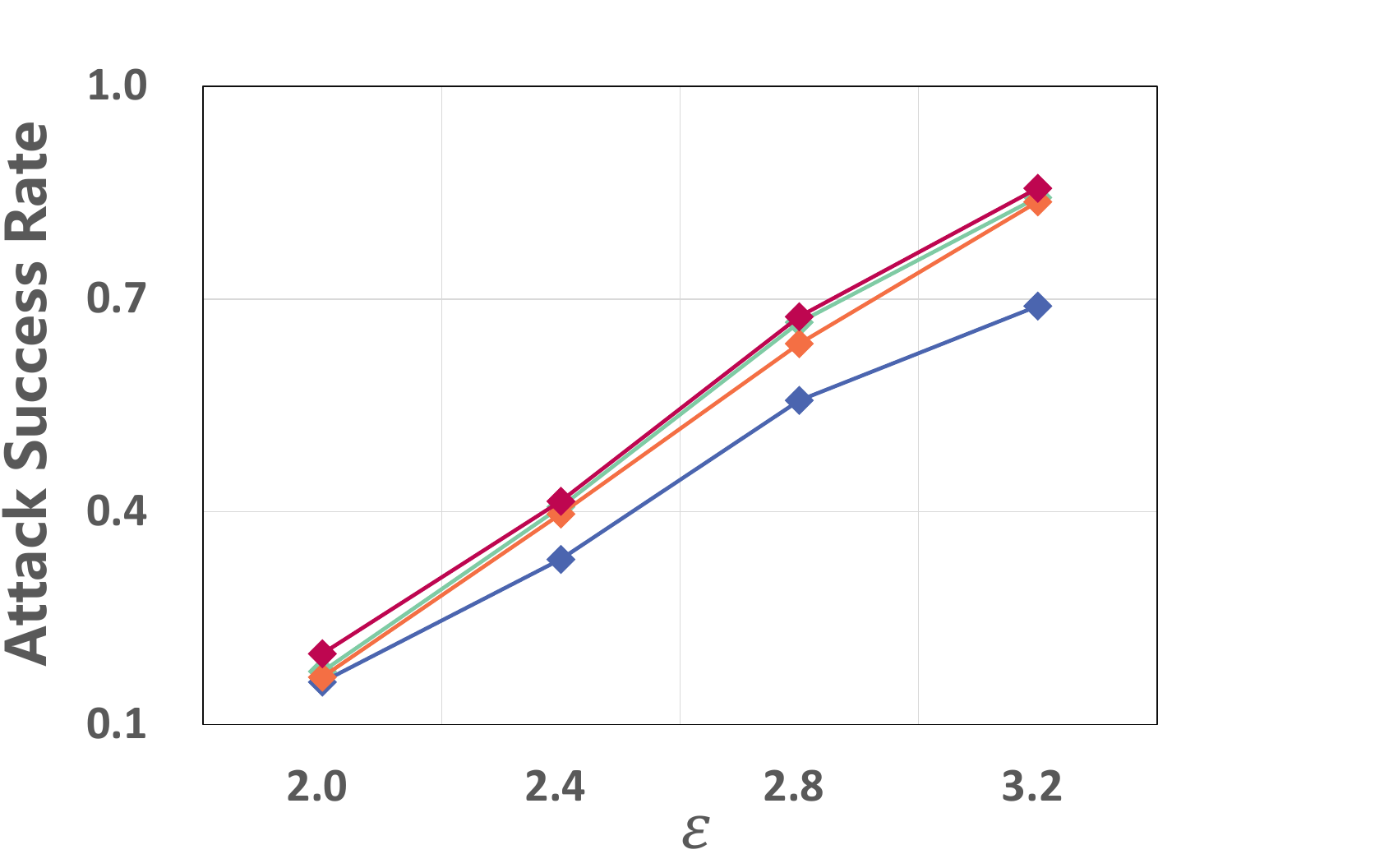}
    } &
    \subfloat[\textit{SANTEXT+}\_QNLI\hspace*{6.5pt}]{
        \includegraphics[width=0.35\textwidth]{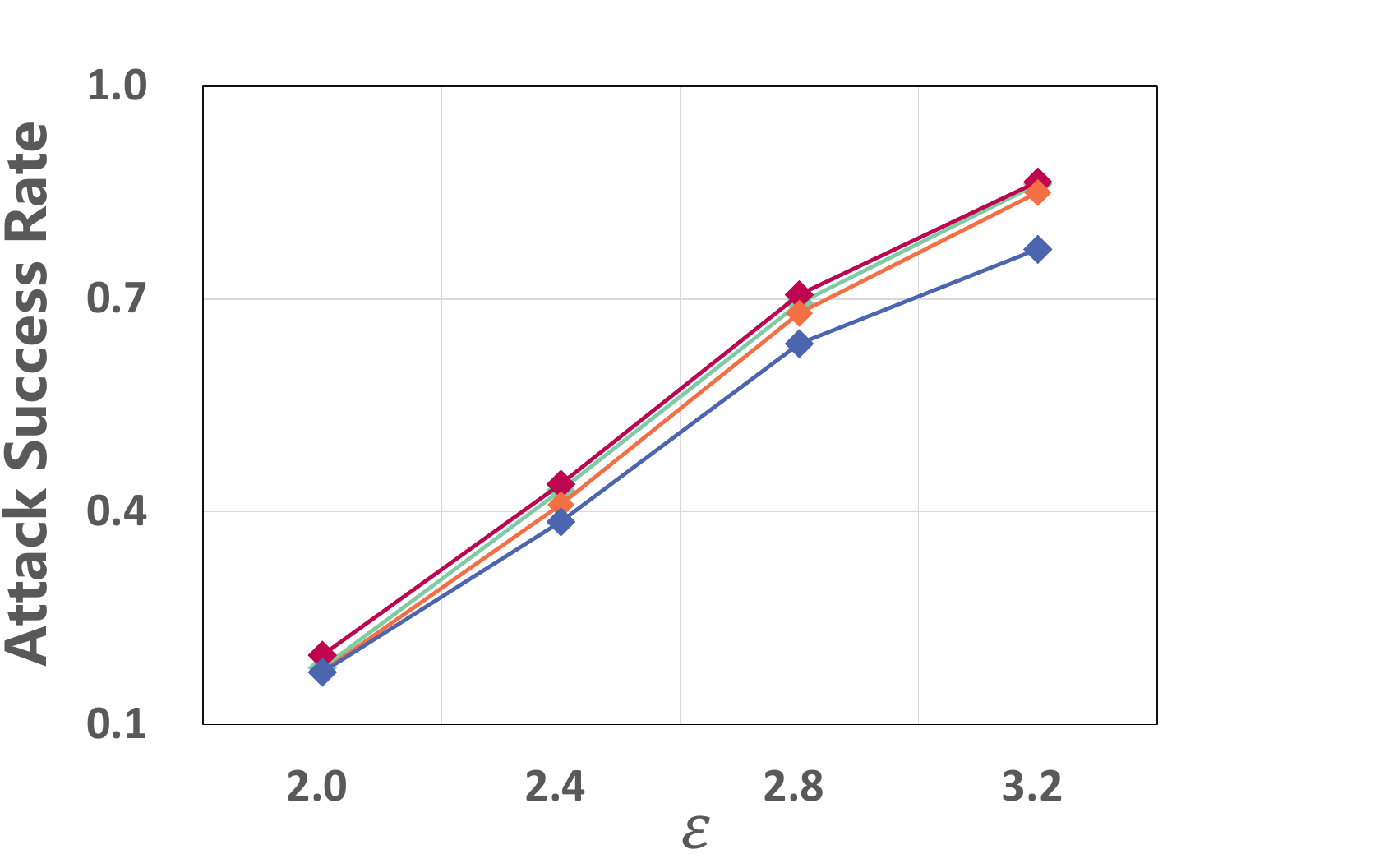}
    } &
    \subfloat[\textit{SANTEXT+}\_AGNEWS\hspace*{8.5pt}]{
        \includegraphics[width=0.35\textwidth]{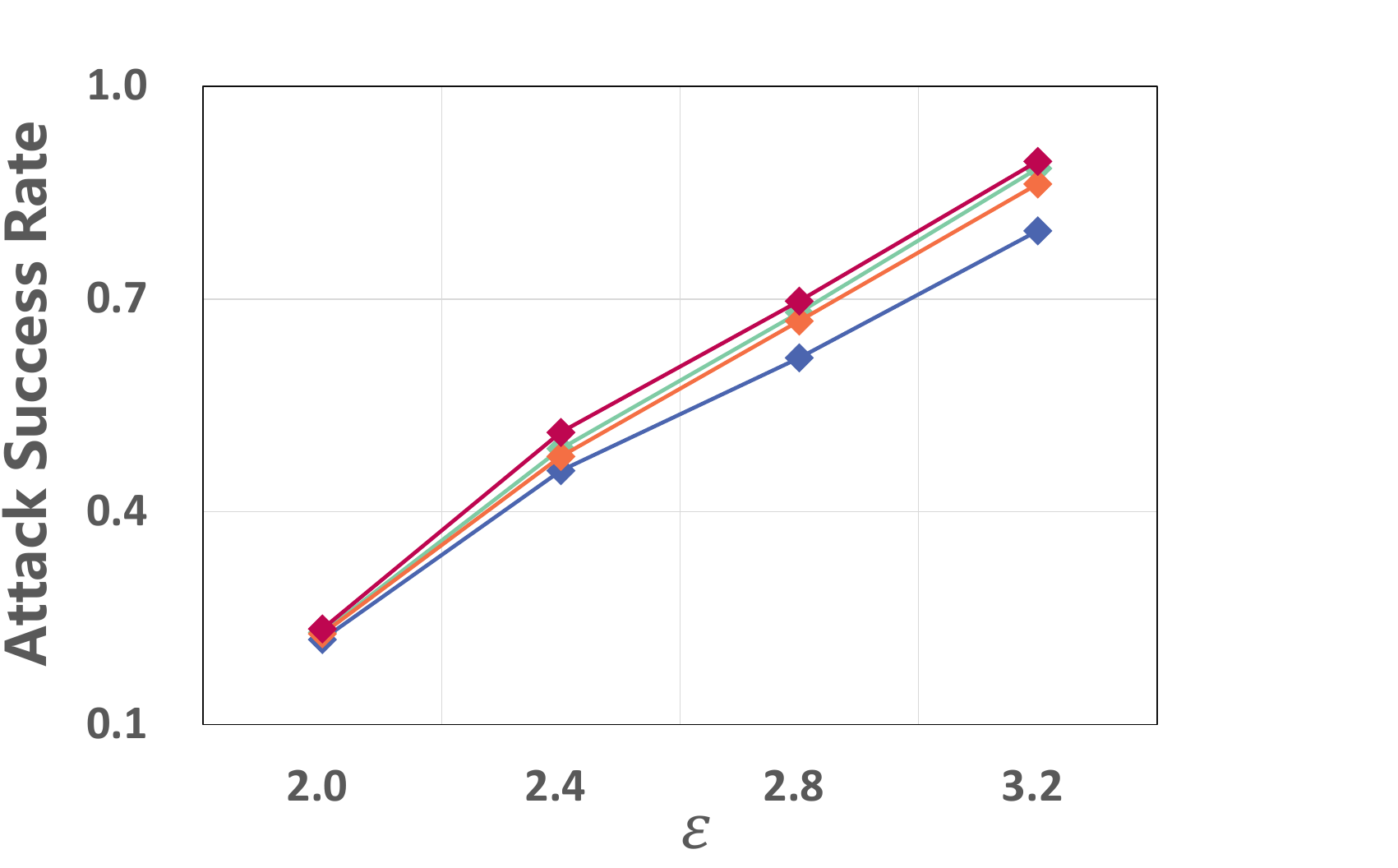}
    }\\
        \multicolumn{3}{c}{\includegraphics[width=1\textwidth]{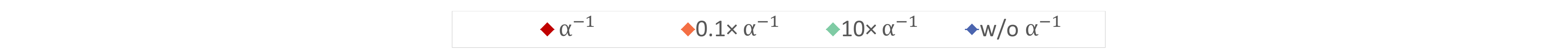}}
    \end{tabular}
    \caption{ASR of \textit{Context-free Bayesian Attack} using various constants instead of $\alpha^{-1}$  on SST-2 dataset.}
    \label{DS}
\end{figure*}

The \textit{Contextual \textit{K}-ASR Bound} provides an approximation of the ASR bound derived from~\Cref{theorem:2}.  It calculates only the top \(K\) values of \(\Pr({x_{i}^j}' | y_{i}^j,\mathbf{c}_{i}^j)\), sorted by \(\Pr({x_{i}^j}' | y_{i}^j)\). Because traversing the entire ${x_{i}^j}' \in X$ for the attack outcome requires a significant amount of time, as detailed in~\Cref{tmk}. We find that it is not necessary to do so. The experiments in~\Cref{ti2} indicate that the top \(K\) tokens sorted by \(\Pr({x_{i}^j}' | y_{i}^j)\) include the majority of original tokens. Increasing $K$ does not necessarily improve ASR. This method balances the trade-off between time cost and ASR.

For probability computations in the \textit{Contextual \textit{K}-ASR Bound}, it easily calculates $\Pr({x_{i}^j}')$ based on the theoretical probability of the original token ${x_{i}^j}'$ appearing in the original sentences. However, given that $\mathbf{c}_i^j$ is a sequence of sanitized tokens, it is challenging to calculate probability $\Pr(\mathbf{c}_{i}^j| {x_{i}^j}', y_{i}^j)$. Although the chain rule can break it down into multiple components~\cite{bain1992introduction}, rebuilding the conditional probabilities for each $\mathbf{c}_{i}^j$ is time-consuming.

To approximate the probability $\Pr(\mathbf{c}_{i}^j| {x_{i}^j}', y_{i}^j)$, we convert it into a binary classification task. We represent it by the prediction probability of a specific label from a BERT model \cite{devlin2018bert} that has been fine-tuned on both the sanitized and original sentences. We detail the implementation of this binary classification in the following~\cref{sec:pattack}. 

Based on the preceding discussion, we have outlined a panoramic view of the theoretically optimal reconstruction attacks and their corresponding ASR bounds. It is important to note that these attacks are developed based on theoretical probabilities. However, practical reconstruction attacks often fall short of this level of precision, typically approximating probability distributions using a shadow dataset. To address this, we introduce the \textit{Context-free Bayesian Attack} and the \textit{Contextual Bayesian Attack} as practical implementations of theoretically optimal reconstruction attacks.

\section{Practical Implementations of Theorems}\label{sec:pattack}

\subsection{Threat Model for Practical Attacks}
\mypara{Adversary’s Goal} The goal of an adversary is to obtain sensitive information by reconstructing each original token from its sanitized token and corresponding sanitized context. 

\mypara{Adversary’s Knowledge} We assume an adversary that accesses the information: 1) \textit{Differentially Private Algorithms.} We assume that an adversary is informed about the differentially private algorithms with its sampling probabilities \(\Pr\left(Y|X\right)\) for text sanitization. Differentially private algorithms are typically treated as public access~\cite{kairouz2014extremal,cormode2018privacy,wang2019collecting,murakami2019utility}. 2) \textit{A Shadow Dataset.} Additionally, we assume that an adversary has access to a shadow dataset, consistent with related privacy attacks~\cite{carlini2022membership,ye2022enhanced}. The shadow dataset is related to but disjointed from the original sentences corresponding to the sanitized tokens under attack.

\subsection{Context-free Bayesian Attack (based on Theorem 1)}
\mypara{Attack Strategy in \textit{Context-free Bayesian Attack}} \textit{Given a sanitized token $y\in Y$, the strategy for Context-free~Bayesian~Attack to reconstruct the original token of $y$ is to select the token ${x}' \in X$ according to the following rule:
\begin{equation}
{x}' = \underset{{x}' \in X}{\mathrm{arg\,max}\,} \frac{\Pr\left(y|{x}'\right) (\Pr\left({x}'\right)+\alpha^{-1})}{\Pr\left(y\right)},
\end{equation}
where $\Pr\left({x}'\right)$ is the statistical probability in the shadow dataset, and $\alpha$ is the number of tokens in the shadow dataset.
}

Unfortunately, the practical implementation of \Cref{theorem:1} encounters difficulty when the approximate probabilities $\Pr\left({x}'\right)$ of all potential original tokens corresponding to $y$ are zero in the shadow dataset. Under these conditions, it is challenging to select an ${x}'$ that uniquely maximizes $\Pr\left({x}'|y\right)$, as these probabilities $\Pr\left({x}'|y\right)$ are all zero. To address this, a small constant, $\alpha^{-1}$, is introduced in the attack. This modification enables \textit{Context-free~Bayesian~Attack} to select ${x}' = \arg\max \Pr\left(y|{x}'\right)$ as the reconstructed token in such cases. Moreover, $\alpha^{-1}$ has little effect on the outcome of the reconstruction attack when the probability $\Pr\left({x}'\right)$ is not zero. 

\Cref{DS} presents the ablation studies on parameter $\alpha^{-1}$ with the \textit{SANTEXT+}~\cite{yue-etal-2021-differential} defense, which demonstrated the effectiveness of $\alpha^{-1}$ in practical reconstruction attacks.
More experiments discussed in \Cref{dsa} further explain why we use $\alpha^{-1}$ in \textit{Context-free Bayesian Attack}. 

\begin{figure*}[t]
    \setlength{\tabcolsep}{-7pt} 
    \centering
    \begin{tabular}{ccc}
    \subfloat[\textit{CUSTEXT+}\_SST-2\hspace*{-7.5pt}]{
        \includegraphics[width=0.35\textwidth]{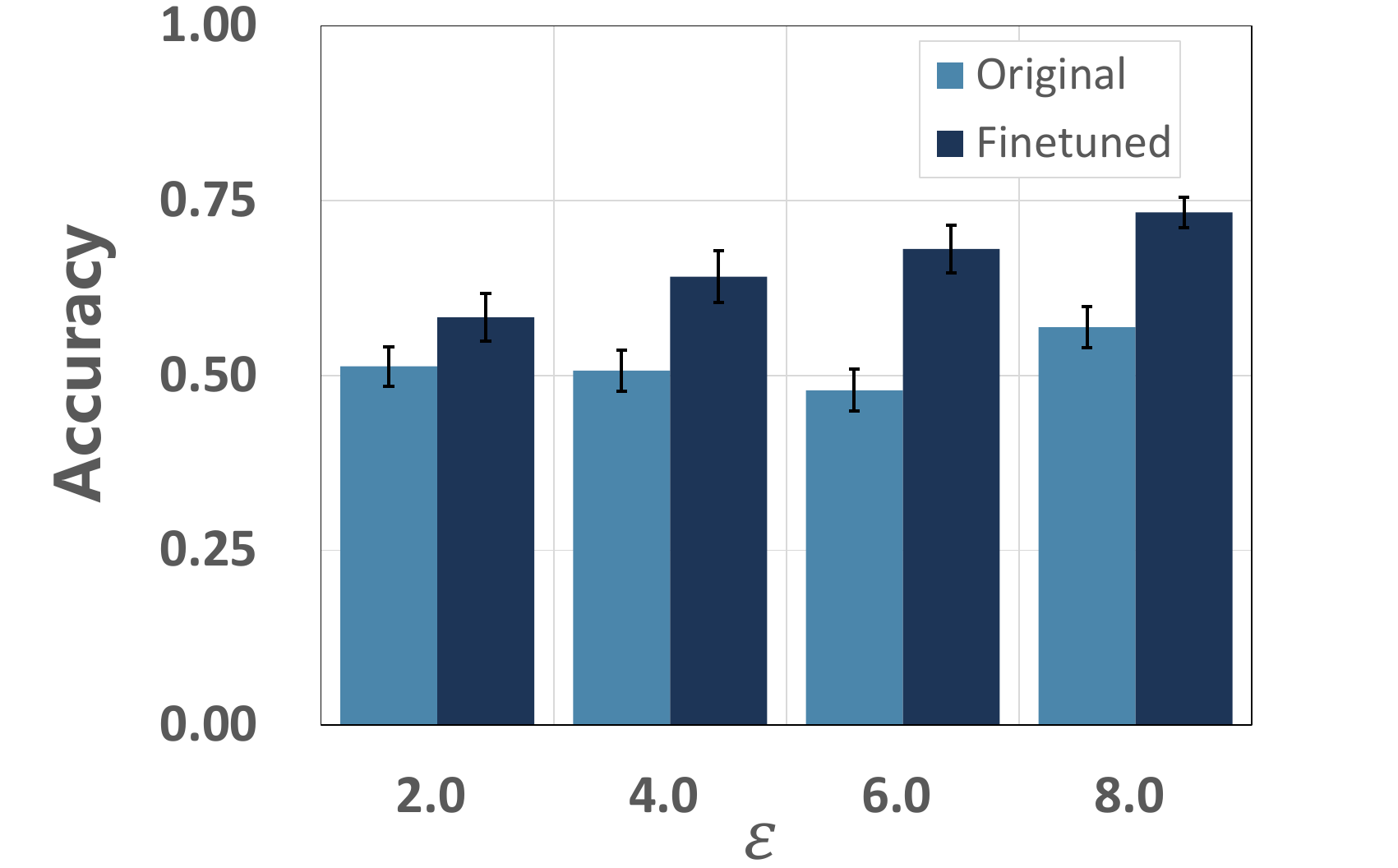}
    } &
    \hspace*{0pt}
    \subfloat[\textit{CUSTEXT+}\_QNLI\hspace*{-8.5pt}]{
        \includegraphics[width=0.35\textwidth]{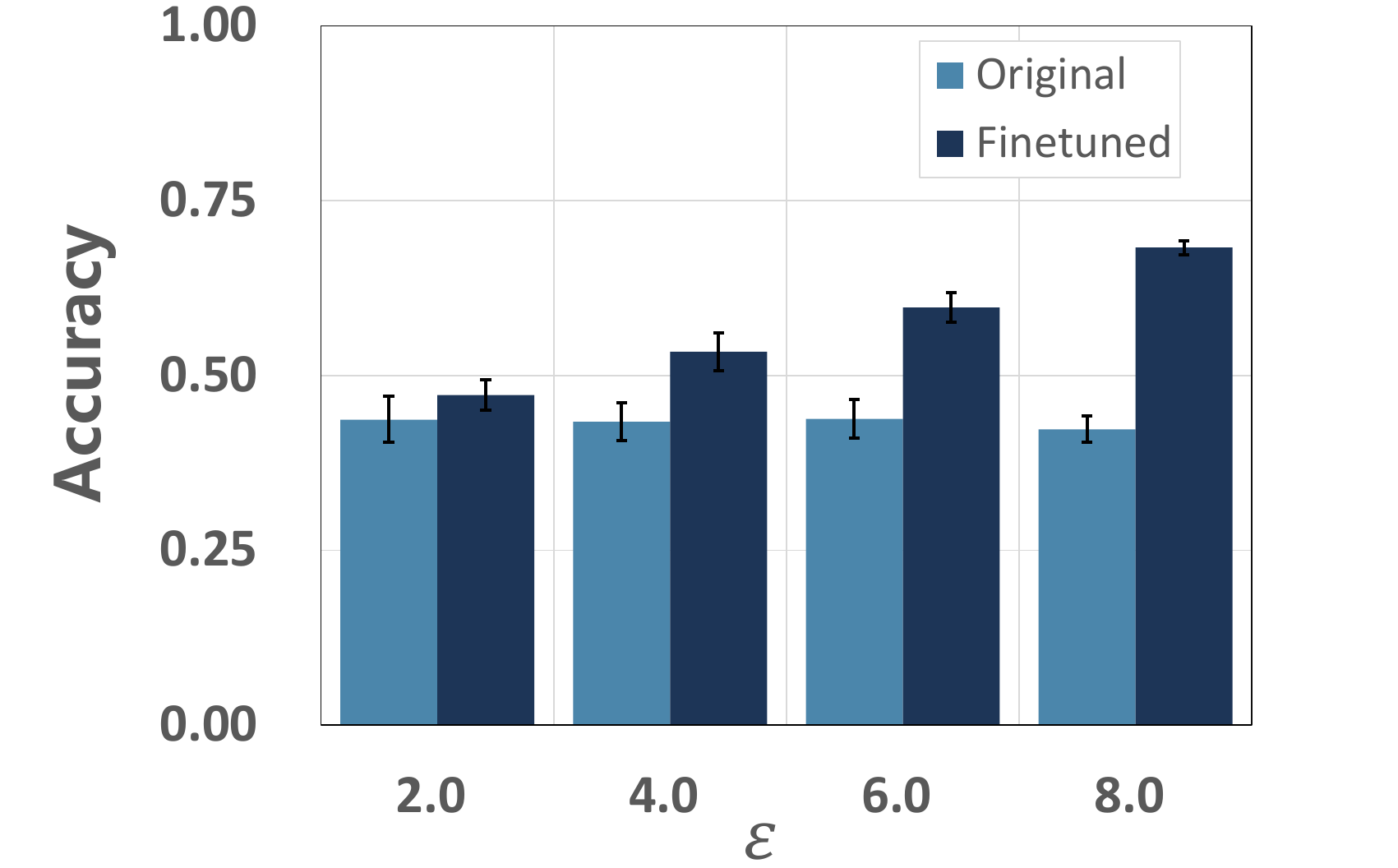}
    } &
    \hspace*{0pt}
    \subfloat[\textit{CUSTEXT+}\_AGNEWS\hspace*{-13.5pt}]{
        \includegraphics[width=0.35\textwidth]{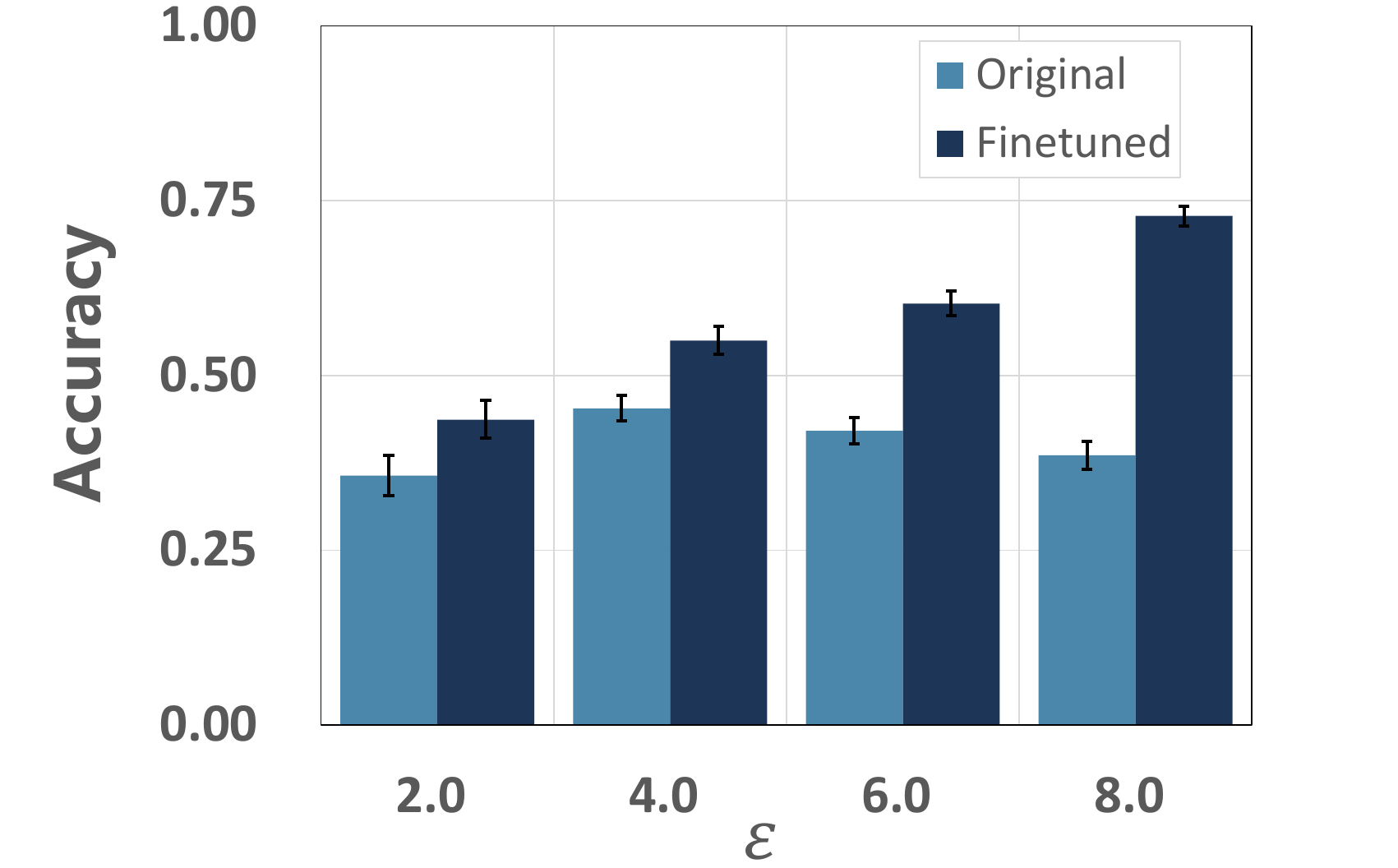}
    }
    \end{tabular}
    \caption{Classification accuracy of detector $h$ in \textit{Contextual Bayesian Attack} with \textit{CUSTEXT+}.}
    \label{fig:2}
\end{figure*}

\subsection{{Contextual Bayesian Attack} (based on Theorem 2)}
\mypara{Attack Strategy in \textit{Contextual Bayesian Attack}} \textit{Given a sanitized token $y$ and its sanitized context $\mathbf{c}$, the strategy for Contextual Bayesian Attack to reconstruct the original token of $y$ is to select the token ${x}'\in X$ according to the following rule:
\begin{equation*}
{x}' = \underset{{x}' \in X'}{\mathrm{arg\,max\,}} \frac{\Pr\left(y|{x}'\right) \cdot (\Pr\left({x}'\right)+\alpha^{-1})}{\Pr\left(y\right)} \times 
\end{equation*}
\begin{equation}\label{al1}
\frac{\Pr\left(\mathbf{c}| {x}', y\right)}{\Pr\left(\mathbf{c}|y\right)},
\end{equation}
\begin{equation}
X' = \underset{{x}' \in X}{\mathrm{arg\,top\,}K\,} \frac{\Pr\left(y|{x}'\right) \cdot (\Pr\left({x}'\right)+\alpha^{-1})}{\Pr(y)},
\end{equation}
where $\Pr\left({x}'\right)$ is the statistical probability in shadow dataset; $\alpha$ is the count of tokens in shadow dataset; $\Pr\left(\mathbf{c}| {x}', y\right)$ is obtained from a fine-tuned model on shadow dataset.}

\mypara{Transforming\;Reconstruction\;to\;Classification} 
To represent the probability $\Pr\left(\mathbf{c} | {x}', y\right)$, we transform it with constant tokens ${x}'$ and $y$: $\Pr\left(\mathbf{c} | {x}', y\right)=\Pr\left(f\left(\mathbf{c}, {x}'\right), f\left(\mathbf{c},y \right) | {x}', y\right)$,~where the function $f$ is defined as follows:

\textit{Given any \( \mathbf{c} = y_{1}, y_{2},  \ldots, y_{i-1}, y_{i+1}, \ldots, y_n \), let $f\left(\mathbf{c}, {x}'\right) = y_1, y_2, \ldots, y_{i-1}, {x}', y_{i+1}, \ldots, y_n$ and $f\left(\mathbf{c}, y\right)= y_1, y_2, \ldots, y_{i-1}, y, y_{i+1}, \ldots, y_n$.}

To calculate $\Pr\left(f\left(\mathbf{c}, {x}'\right), f\left(\mathbf{c},y \right) | {x}', y\right)$, we first fine-tune a BERT model \cite{devlin2018bert} to serve as a detector, denoted by $h$. This detector $h$ is a binary classifier that infers whether $f\left(\mathbf{c},{x}'\right)$ and $f\left(\mathbf{c}, y\right)$ correspond to ${x}'$ and $y$ in sanitization process (i.e., $\mathbf{c}$ is the sanitized context for ${x}'$ and $y$). Specifically, the input to detector $h$ consists of two token sequences, $f\left(\mathbf{c},{x}' \right)$ and $f\left(\mathbf{c}, y \right)$. The output of detector $h$ is either label $1$ or label $0$, where a value of $1$ indicates that $f\left(\mathbf{c},{x}' \right)$ and $f\left(\mathbf{c}, y \right)$ correspond to ${x}'$ and $y$, while a value of $0$ indicates that they do not. With these inputs and outputs, $h$ is defined as $h\left(f\left(\mathbf{c}, y\right), f\left(\mathbf{c}, {x}'\right)\right) \rightarrow \{0, 1\}$. The training samples of $h$ are constructed as follows:

\mypara{\textit{Step 1: Confirm the Attack Target}} We first confirm the attack target \(\Gamma\) in order to construct corresponding training samples. The attack target \(\Gamma\) comprises a set of sanitized tokens that need to be reconstructed to their original forms.

\mypara{\textit{Step 2: Construct the Training Samples}} We sanitize sentences in the shadow dataset using the differentially private algorithm. For each token \(y \in \Gamma\) in the sanitized sentences, we create a training sample consisting of two token sequences and a label for $y$. These two token sequences are the input of detector $h$: The first sequence includes the sanitized sentence containing \(y\) and \(\mathbf{c}\). The second sequence replaces the token \(y \in \Gamma\) in the first sequence with its reconstructed result \({x}'\) from the \textit{Context-free Bayesian Attack}. The label is the output of detector $h$: If \({x}'\) matches the original token of \(y\), indicating that $f\left(\mathbf{c},{x}'\right)$ and $f\left(\mathbf{c}, y\right)$ correspond to ${x}'$ and $y$ in sanitization process, we label the sample \(1\). Otherwise, we label it \(0\). To maintain label balance, for each sample labeled as \(1\), we generate a corresponding sample labeled as \(0\) using the same sanitized sentence but with the second most likely reconstructed token from the \textit{Context-free Bayesian Attack}. Similarly, for each sample labeled \(0\), we generate a sample labeled \(1\).

With the constructed training samples, we fine-tune the BERT model over three epochs with \textit{Adam\,Optimizer} \cite{kingma2014adam}, obtaining the detector $h$. We represent $\Pr\left(\mathbf{c} | {x}', y\right)=\Pr\left(f\left(\mathbf{c}, {x}'\right), f\left(\mathbf{c},y \right) | {x}', y\right)$ using prediction probability of $h$ for label $1$.

\Cref{fig:2} shows the classification accuracy of $h$, as detailed in~\Cref{clc}. Specifically, ``Original'' in \Cref{fig:2} refers to the performance of directly using BERT without any fine-tuning, while ``Finetuned'' refers to the performance after fine-tuning BERT on the detection task. Experimental results demonstrate the effectiveness of $h$ in determining whether $\mathbf{c}$ is the sanitized context of ${x}'$ and $y$.

\section{Experiments}\label{sec:exp}



\begin{figure*}[t!]
    \setlength{\tabcolsep}{-12pt} 
    \begin{tabular}{ccc}\hspace{8pt}
    \subfloat[\textit{CUSTEXT+}\_SST-2]{\label{main:cus_sst2}
        \includegraphics[width=0.37\textwidth]{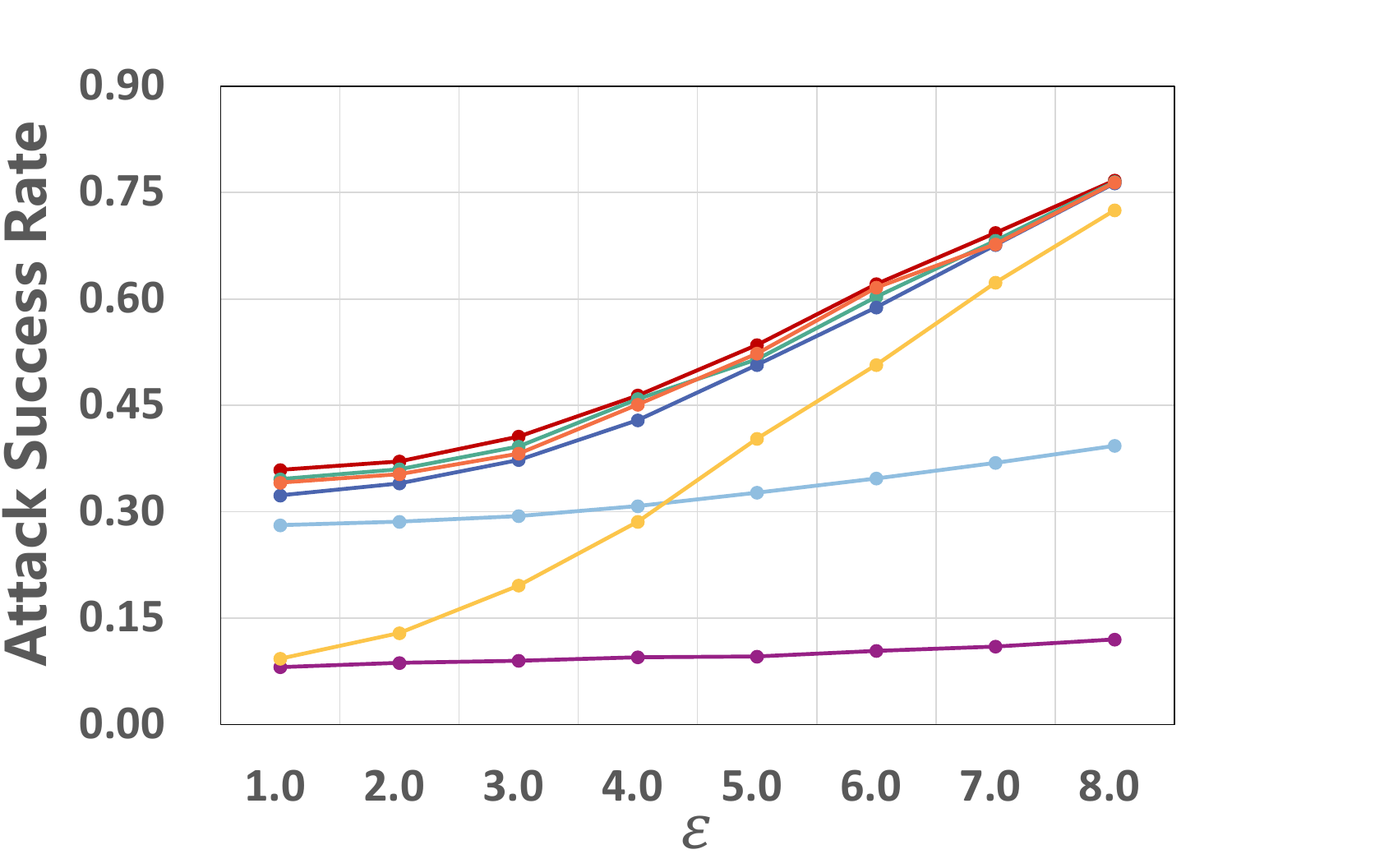}
    } & 
    \subfloat[\textit{CUSTEXT+}\_QNLI]{\label{main:cus_qnli}
        \includegraphics[width=0.37\textwidth]{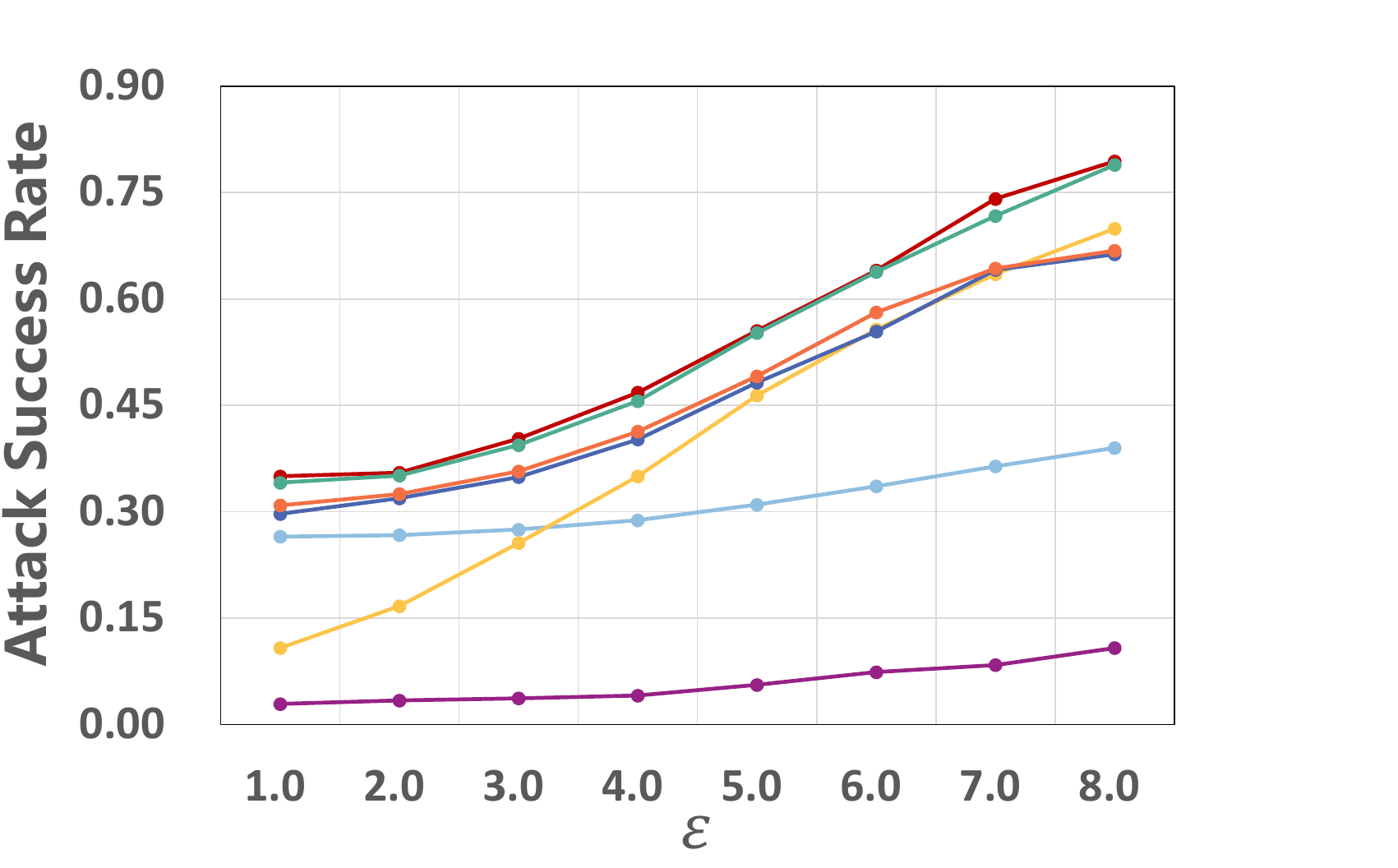}
    } &
    \subfloat[\textit{CUSTEXT+}\_AGNEWS]{\label{main:cus_agnews}
        \includegraphics[width=0.37\textwidth]{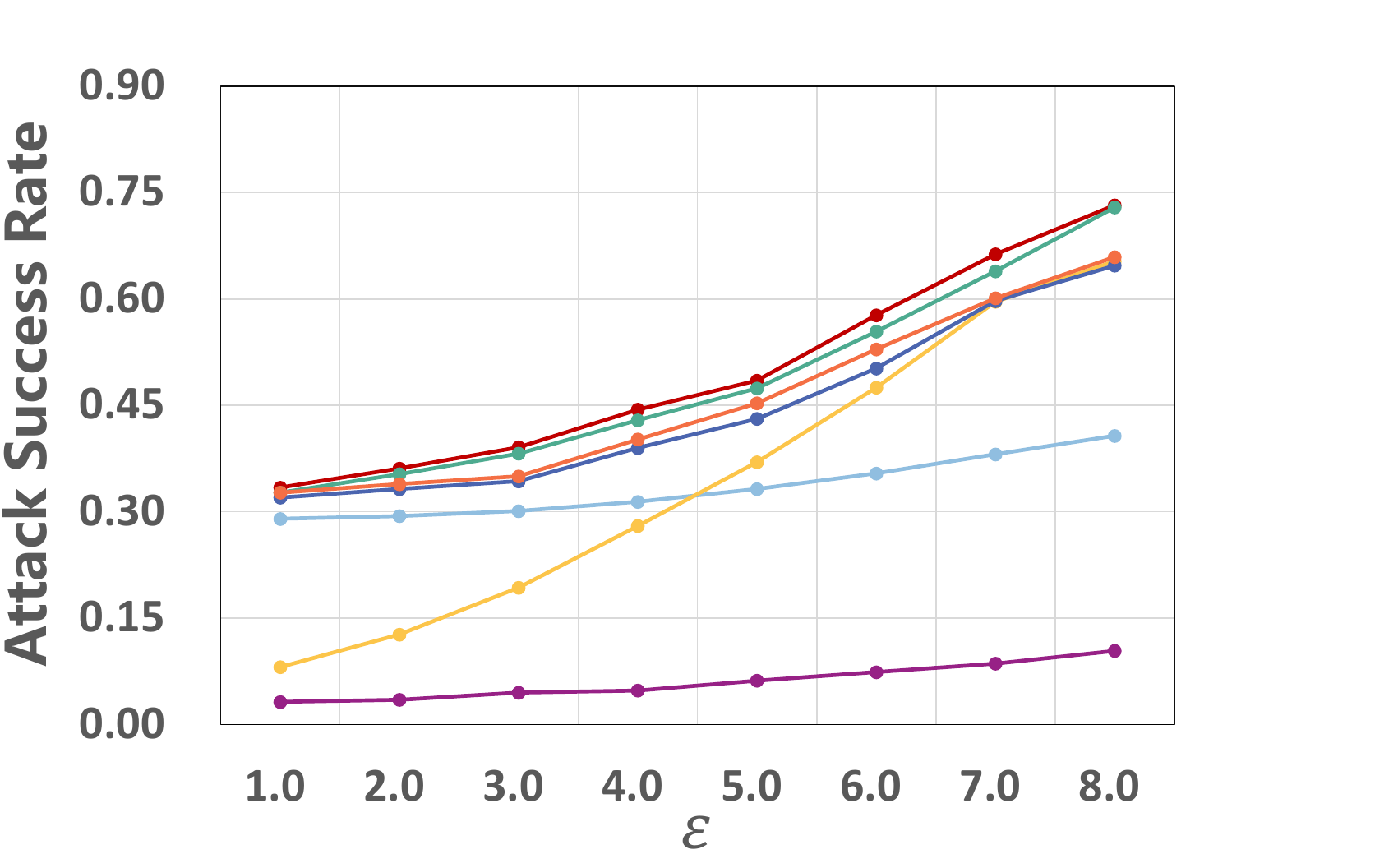}
    }\\    \hspace{8pt}
   \subfloat[\textit{SANTEXT+}\_SST-2]{\label{main:san_sst2}
        \includegraphics[width=0.37\textwidth]{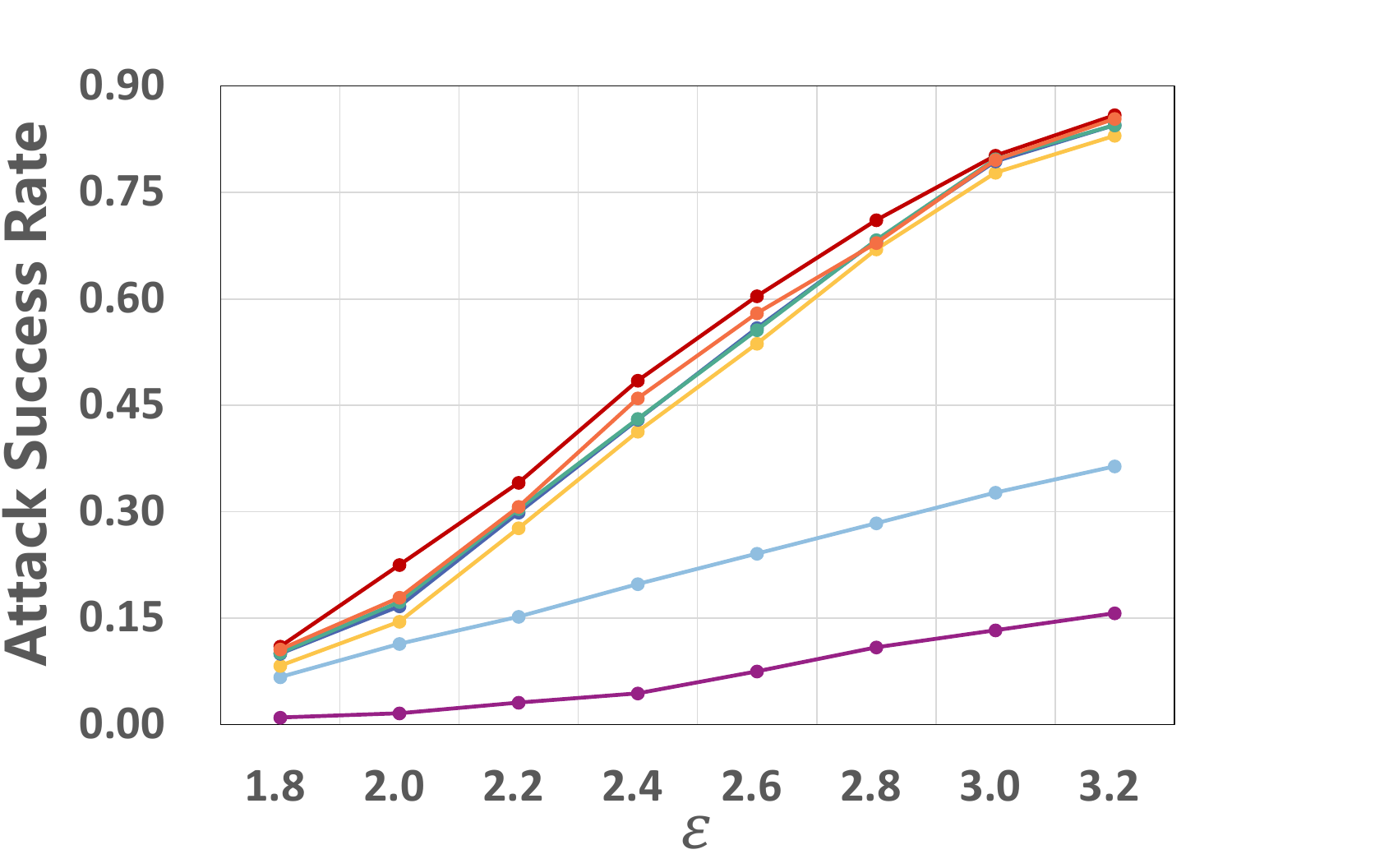}
    } & 
    \subfloat[\textit{SANTEXT+}\_QNLI]{\label{main:san_qnli}
        \includegraphics[width=0.37\textwidth]{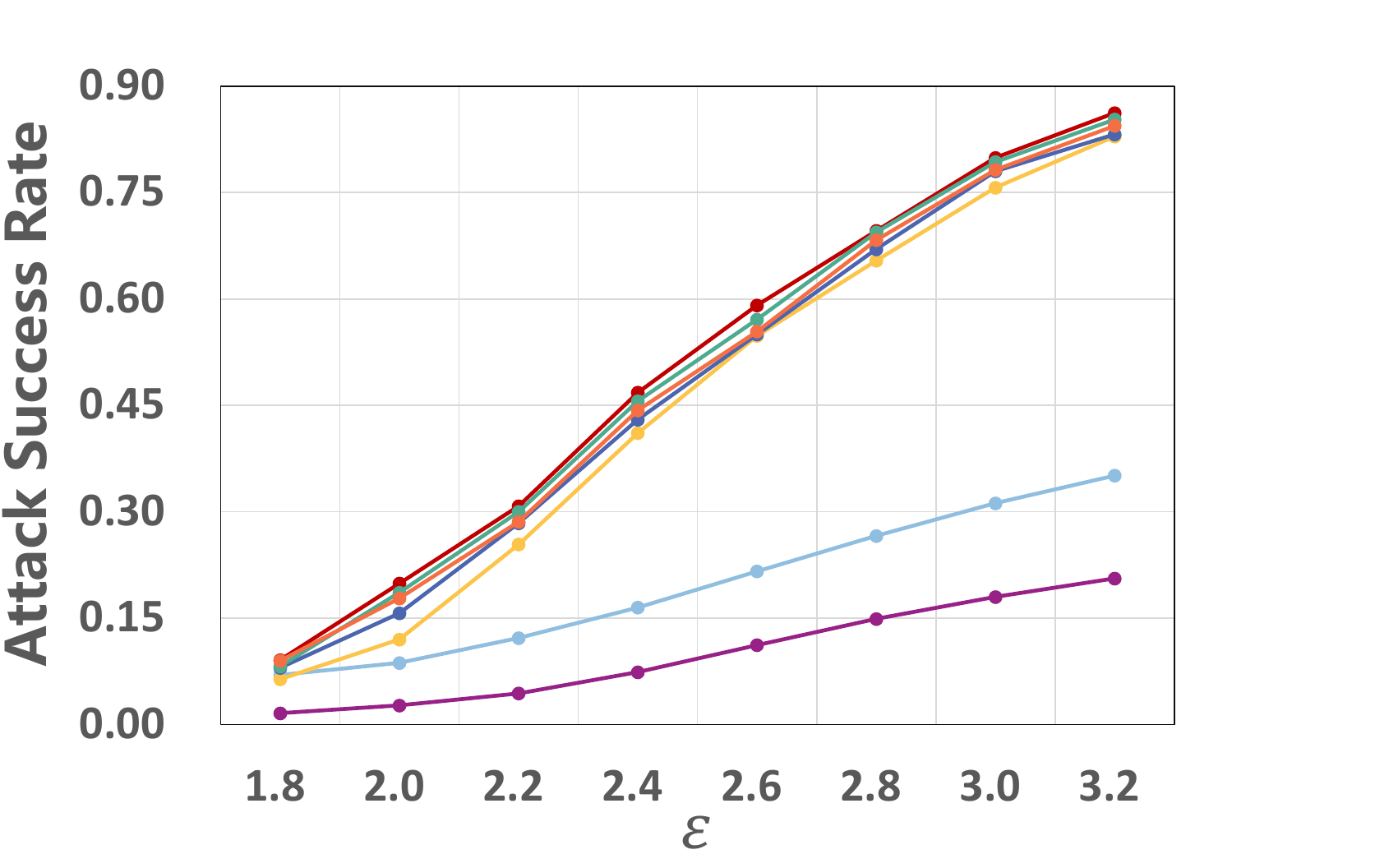}
    } &
    \subfloat[\textit{SANTEXT+}\_AGNEWS]{
        \includegraphics[width=0.37\textwidth]{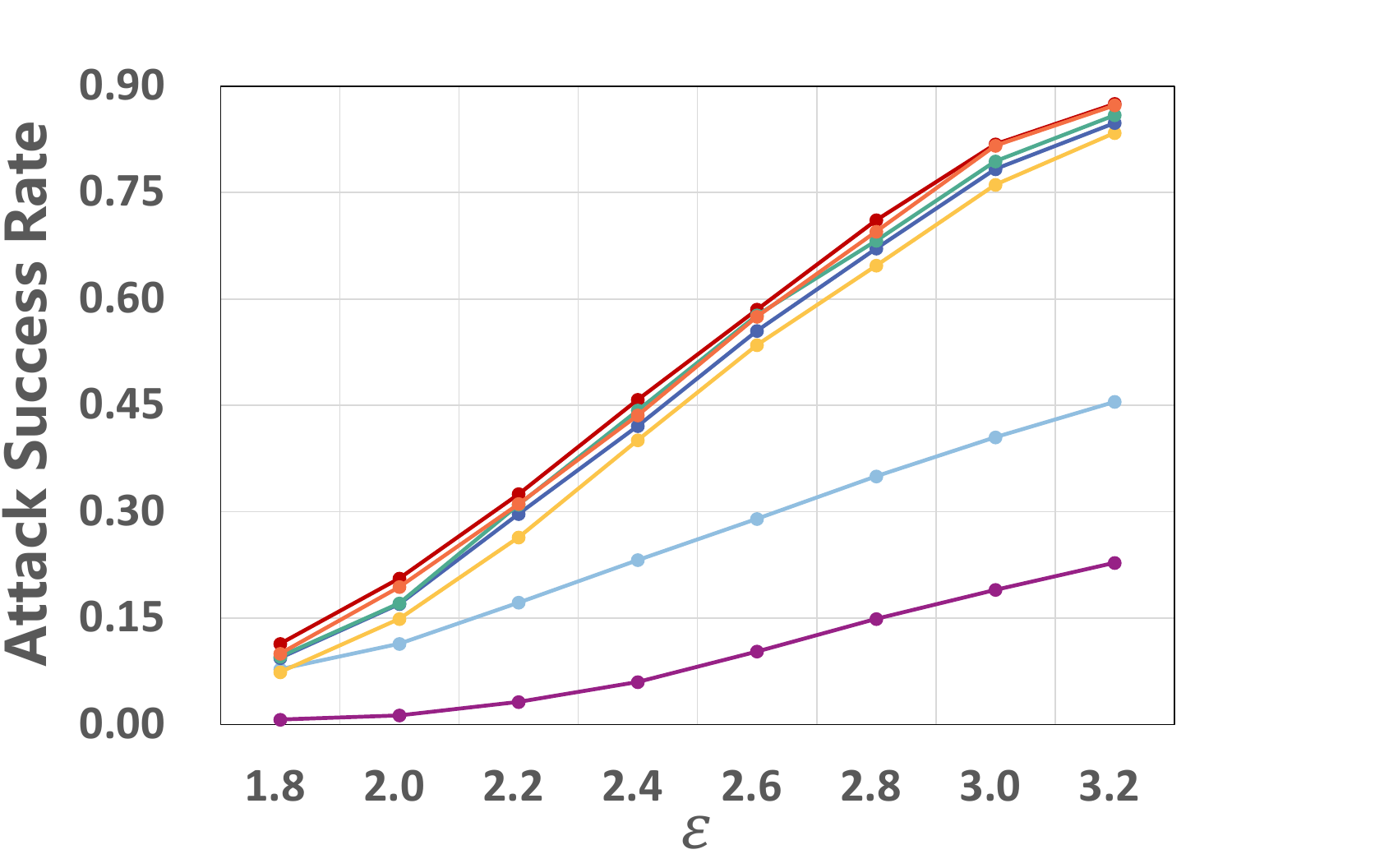}
    }
    \end{tabular}\\
    \centering
    \hspace{2pt}
    \includegraphics[width=1\textwidth]{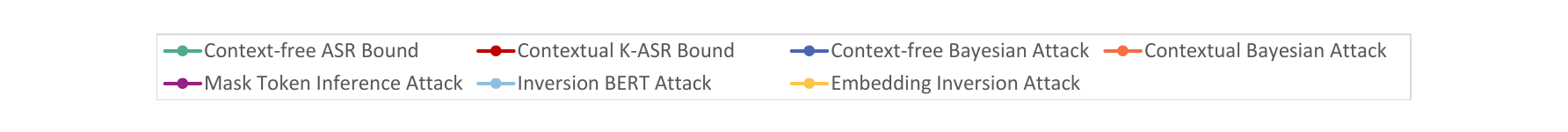}
    \vspace{-0.3cm}
    \caption{ASR of reconstructing sensitive tokens against \textit{CUSTEXT+} and \textit{SANTEXT+} with $K=10$.}
    \label{tab:attack-result}
\end{figure*}

\subsection{Experimental Setup}
\mypara{Datasets} To validate the effectiveness of our attacks, we conduct experiments on four datasets, SST-2~\cite{socher-etal-2013-recursive}, AGNEWS~\cite{2zhang2015character}, QNLI~\cite{wang2018glue}, and Yelp~\cite{2zhang2015character}. These datasets cover a wide range of sensitive entities and personally identifiable information (PII) (detailed in~\Cref{dats}). 

\mypara{Defense Methods} We implement reconstruction attacks against two backbone methods of text sanitization using their default settings for sensitive tokens and parameters: \textit{CUSTEXT+}~\cite{chen2023customized} and \textit{SANTEXT+}~\cite{yue-etal-2021-differential}, both of which are widely used in DP research. 

\mypara{Baseline Attacks} Our work presents the first systematic study of reconstruction attacks against text sanitization. Since there is no other research specifically focusing on reconstruction attacks against text sanitization, we compare our methods with three adaptive attacks from prior defenses \cite{tong2023InferDPT,zhou2023textobfuscator}: \textit{Embedding} \textit{Inversion} \textit{Attack}~\cite{qu2021natural}, \textit{Inversion\,BERT\,Attack}~\cite{kugler2021invbert}, and \textit{Mask\,Token\,Inference\,Attack}~\cite{yue-etal-2021-differential}.

\mypara{Implementation} We run experiments on a 
RTX A6000 GPU. We randomly select $20,000$ sentences from the entire dataset to form a private dataset, whose sensitive tokens are perturbed to sanitized tokens. In \textit{CUSTEXT+}, we calculate the ASR of the successful reconstruction for the sensitive tokens from $3,000$ sanitized tokens. In \textit{SANTEXT+}, this calculation of ASR is based on $1,000$ sanitized tokens. For a shadow dataset, we randomly select $1,000$ sentences from the entire dataset, which are disjoint from those to reconstruct. 

More implementation details of our experiments can be found in~\Cref{imp}. 

\subsection{Performance of Reconstruction Attacks}\label{att}

\mypara{Baseline Reconstruction Attacks} In this paper, we identify three factors that facilitate reconstructing original tokens: the sampling probabilities of DP, the probability distribution of input tokens, and the sanitized context. However, the baseline attacks generally fall short in a design that considers all three factors. The limitations of them are demonstrated by the experimental results. For instance, at $\epsilon=3.2$ on the QNLI, as shown in \Cref{main:san_qnli}, both \textit{Inversion BERT Attack} and \textit{Mask Token Inference Attack} successfully reconstructed the sanitized tokens less than half of those in others. Moreover, at $\epsilon=1.0$ on the AGNEWS, as shown in \Cref{main:cus_agnews}, the ASR of \textit{Embedding Inversion Attack} is less than $0.09$, which is approximately one-third of that in \textit{Inversion BERT Attack}.

\mypara{Attack Success Rate Bounds} Experimental results in \Cref{tab:attack-result} indicate that both \textit{Context-free ASR Bound} and \textit{Contextual K-ASR Bound} more accurately evaluate the effectiveness of DP-based text sanitization compared to baselines. {For instance, with a privacy parameter $\epsilon=1.0$ on QNLI, as shown in \Cref{main:cus_qnli}, the ASR of the \textit{Context-free ASR Bound} is $28.7\%$ higher than that of the \textit{Inversion BERT Attack}, while the ASR of the \textit{Contextual K-ASR Bound} is $32.1\%$ higher than that of the \textit{Inversion BERT Attack}.} Furthermore, our experiments demonstrate that the \textit{Contextual K-ASR Bound} is consistently higher than the \textit{Context-free ASR Bound}. Specifically, as shown in \Cref{main:san_sst2}, the \textit{Contextual K-ASR Bound} exceeds the \textit{Context-free ASR Bound} by 30.1\% with $\epsilon=2.0$ on SST-2. 

\mypara{Practical Attacks from Theorems} According to \Cref{tab:attack-result}, the ASR of both the \textit{Context-free Bayesian Attack} and the \textit{Contextual Bayesian Attack} exceeded that of the baseline attacks in most cases. Notably, the \textit{Contextual Bayesian Attack} achieves a 46.4\% improvement in ASR over the \textit{Inversion BERT Attack} at a privacy budget of $\epsilon=4.0$ on the SST-2 according to \Cref{main:cus_sst2}. With the same experimental setting, the \textit{Context-free Bayesian Attack} achieves a 39.2\% improvement in ASR over the \textit{Inversion BERT Attack}. Furthermore, the ASRs of our proposed practical attacks are approaching their corresponding theoretical ASR bounds. For example, with $\epsilon=1.8$ on SST-2 as shown in \Cref{main:san_sst2}, the ASR of the \textit{Context-free Bayesian Attack} was only $0.001$ less than that of the \textit{Context-free ASR Bound}. With the same setting, the ASR of the \textit{Contextual Bayesian Attack} was $0.004$ less than that of the \textit{Contextual K-ASR Bound}. {More discussion on the performance of reconstruction attacks with large $\epsilon$ can be found in \Cref{ct24}.}

Furthermore, our reconstruction attacks remain comparable to the baseline methods even with a misaligned shadow dataset. The detailed experiments can be found in \Cref{ct23}.

 \begin{figure}[t]
    \setlength{\tabcolsep}{-4pt} 
    \centering
    \vspace{0cm}
    \begin{tabular}{cc}\hspace{-0.37cm}\hspace{7pt}
    \subfloat[\textit{CUSTEXT+}\_Yelp]{\label{pii:cus}
        \includegraphics[width=0.29\textwidth]{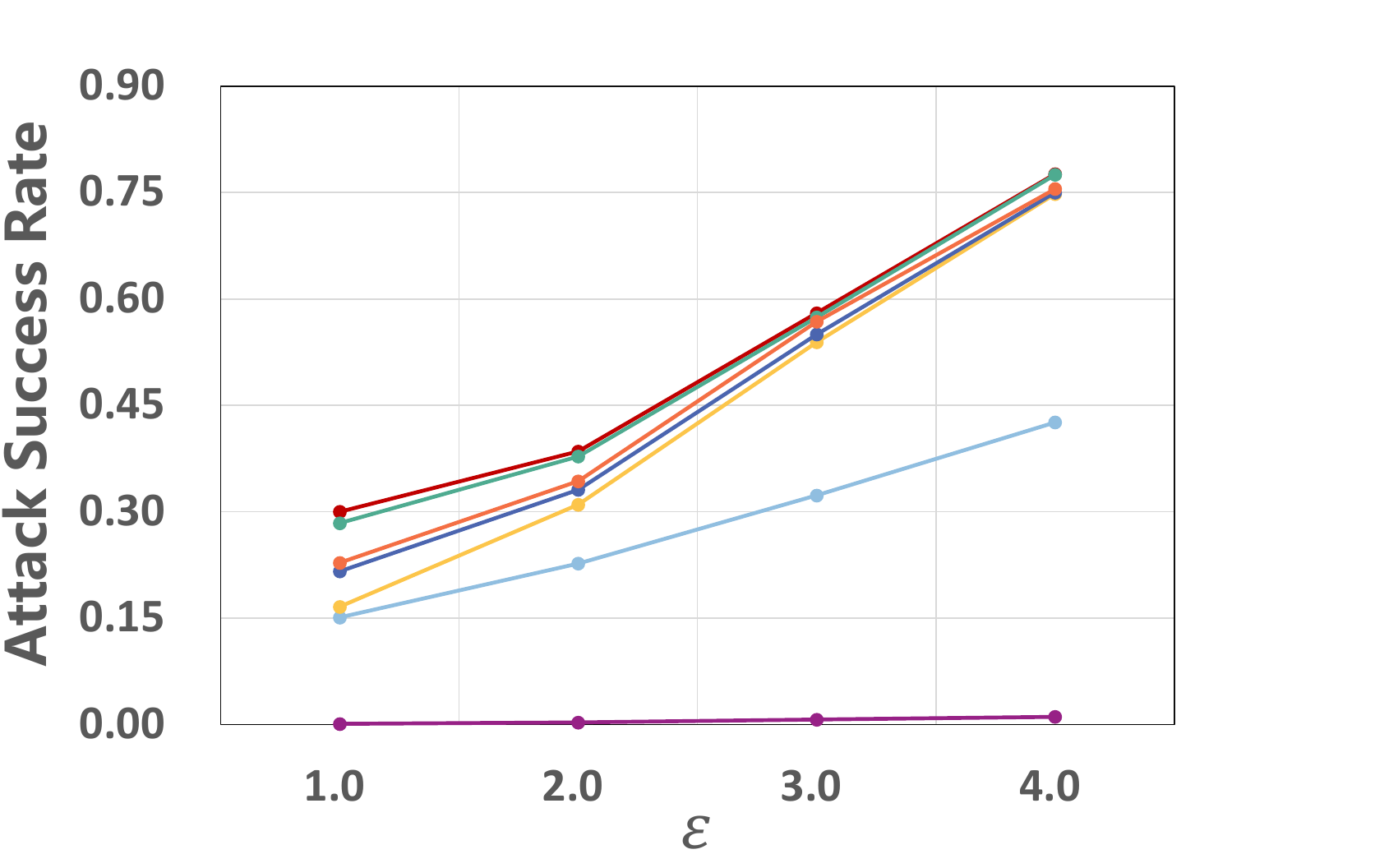}
    } &\hspace{-0.87cm}\hspace{7pt}
    \subfloat[\textit{SANTEXT+}\_Yelp]{
        \includegraphics[width=0.29\textwidth]{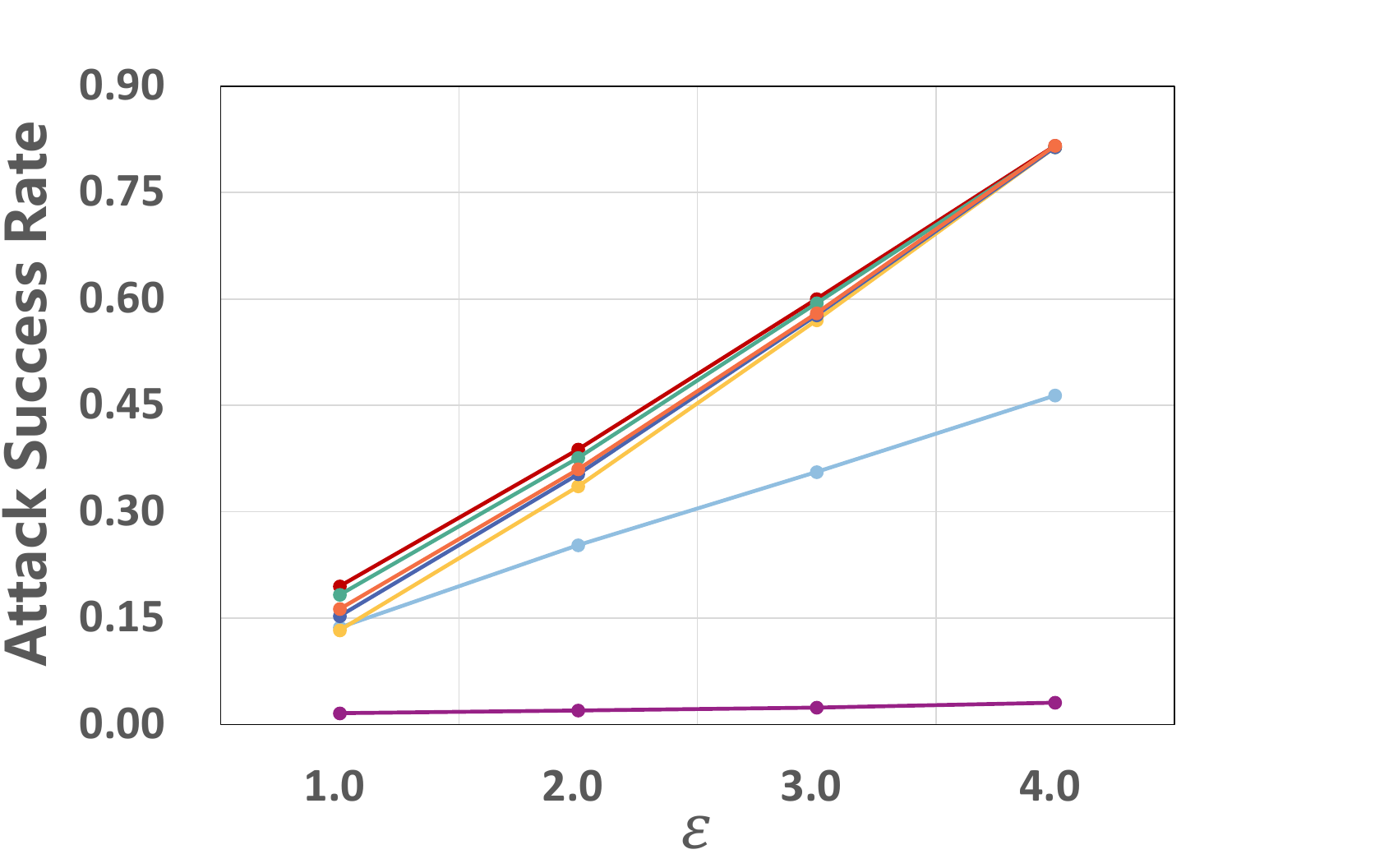}
    }\\
        \multicolumn{2}{c}{\hspace{-0.5cm}\includegraphics[width=1.2\columnwidth]{fig/main_t.pdf}}
    \end{tabular}
    \caption{ASR of reconstructing PII tokens.}
    \label{pii}
\end{figure}
\begin{table}[t]
\setlength{\tabcolsep}{1.8mm} 
\centering
\caption{ASR of reconstruction attacks\;with\;various\;$K$.}
\label{kvalue}
\resizebox{\columnwidth}{!}{
\begin{tabular}{l|c|c|c|c|c}
\toprule
\small\textbf{Attack Strategy}  & \small\textbf{K=1}  & \small\textbf{K=5} & \small\textbf{K=10} & \small\textbf{K=15} & \small\textbf{K=20}   \\ 
\midrule
\small {Contextual Bound}&\small 0.499 &\small  0.520 & \small 0.525 &\small  0.525 &\small  0.525\\ 
\small {Contextual Attack}&\small 0.470  &\small  0.497 &\small  0.503 & \small 0.503 &\small  0.503\\ 
\bottomrule
\end{tabular}
}
\end{table}
\subsection{Reconstructing PII Tokens}
Text sanitization is widely adopted in de-identifying personal information \cite{amazon,Priva}. 
To bring the evaluations closer to real-life scenarios, we conducted reconstruction attacks on the Yelp dataset to recover the PII of personal names \cite{Presidio} from sanitized sentences. Experimental results, as illustrated in \Cref{pii}, indicate the improvement of our attacks over the baselines. For instance, as shown in \Cref{pii:cus}, the \textit{Context-free ASR Bound} increased the ASR by $41.5\%$ with a privacy parameter $\epsilon=1.0$ over \textit{Embedding Inversion Attack}; at the same experimental setting, \textit{Contextual K-ASR Bound} increased ASR by $96.1\%$ over \textit{Inversion BERT Attack}.

\subsection{Time Cost and Parameter Influence}\label{ti2}

\mypara{Time Costs} We estimated the time required for reconstructing $1,000$ sanitized tokens. 
We conducted this experiment $10$ times to calculate standard deviations (1-sigma) using the closed form formula~\cite{altman2005standard}. 
As depicted in \Cref{time2}, computations for the \textit{Context-free ASR Bound} and the \textit{Context-free Bayesian Attack} were completed in less than $5$ minutes. Both of them enable quick evaluation of the effectiveness of text sanitization. For a comprehensive evaluation, \textit{Contextual Bayesian Attack} and \textit{Contextual K-ASR Bound} completed reconstructions in less than $25$ minutes. We further compare the time costs of various reconstruction attacks in \Cref{ct22}.

\begin{figure}[t]
\vspace{0.5cm}
\hspace{1.3cm}
\includegraphics[width=0.64\columnwidth]{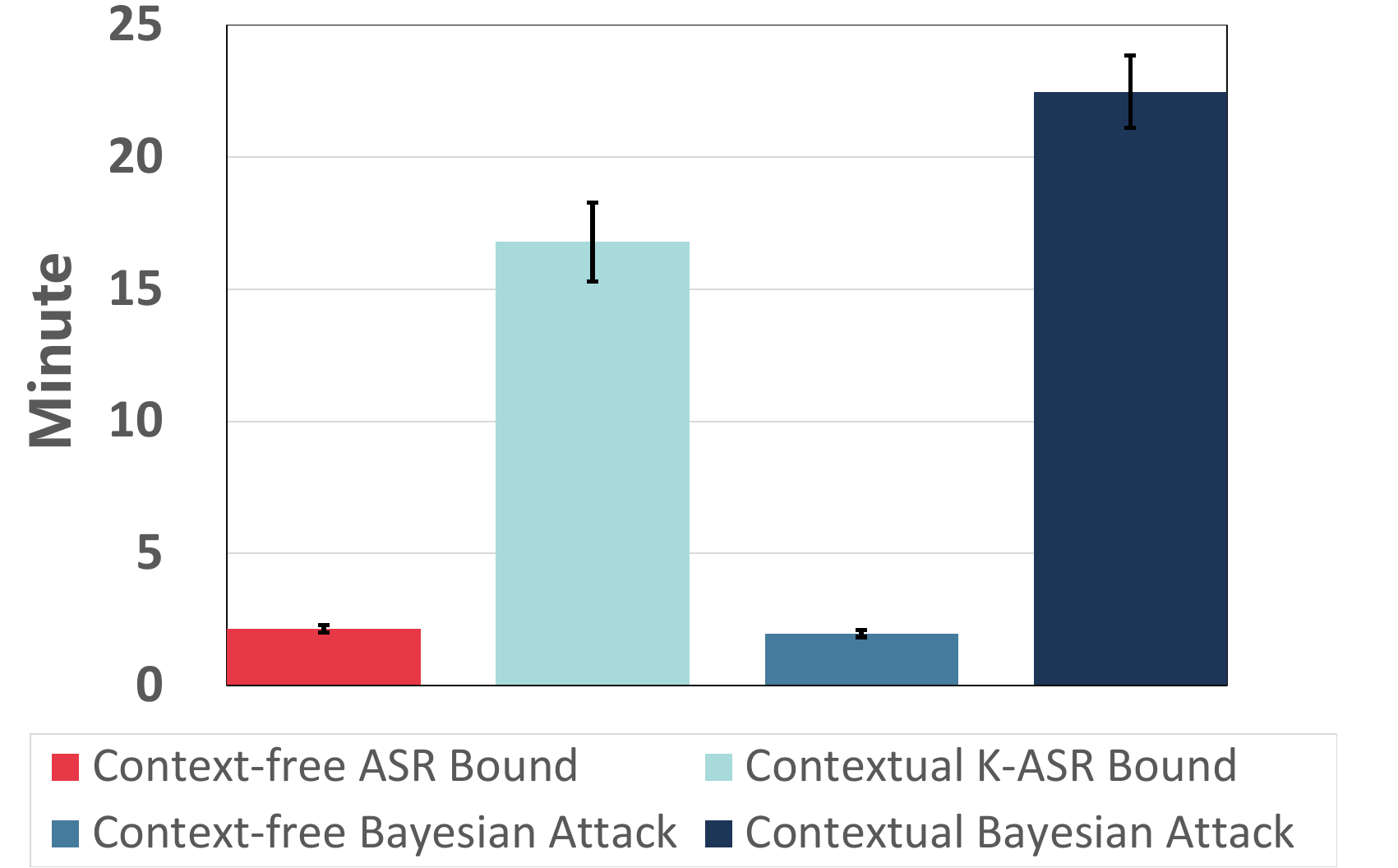}
\caption{Time costs of reconstructing $1,000$ tokens.}
\label{time2}
\end{figure}
\begin{table}[t]
\setlength{\tabcolsep}{1mm} 
\centering
\caption{ASR of reconstructing sensitive tokens with different shadow(S) / private(P) dataset size ratios.}
\label{sd}
\resizebox{\columnwidth}{!}{
\begin{tabular}{l|c|c|c|c}
\toprule
\small \textbf{Attack Strategy} & \small \textbf{{S/P=1\%}} & \small \textbf{{S/P=3\%}} & \small \textbf{{S/P=5\%}} & \small \textbf{{S/P=8\%}} \\ 
\midrule
\small Context-free\,Attack& \small 0.492  & \small 0.494 & \small 0.496 & \small 0.506 \\ 
\small Contextual\,Attack& \small 0.473   & \small 0.487 & \small 0.505 & \small 0.508 \\ 
\bottomrule
\end{tabular}}
\end{table}

\mypara{Parameter Influences} Additionally, we test the influence of parameters on our attacks using SST-2 and \textit{CUSTEXT+} at $\epsilon=4.0$. \Cref{kvalue} shows the influence of $K$ in both \textit{Contextual Bayesian Attack} ({Contextual Attack}) and \textit{Contextual K-ASR Bound} ({Contextual Bound}). When $K$ is smaller than $10$, the ASR of both grows with $K$. When $K$ exceeds $10$, the growth of ASR tends to flatten. \Cref{sd} shows the ASRs for different size ratios between the shadow dataset and the private dataset. With a smaller ratio, the ASR of \textit{Contextual Bayesian Attack} ({Contextual Attack}) and \textit{Context-free Bayesian Attack} ({Context-free Attack}) decreases.



\section{Conclusion}\label{cln}
In this work, we implement both context-free and contextual optimal reconstruction attacks against text sanitization. We derive their bounds on ASR to guide the evaluation of sanitization effectiveness. We propose two practical attacks for real-world implementations. Evaluation results validate that our attacks surpass baseline methods.

This work represents an initial attempt to provide a more precise evaluation of sanitization effectiveness compared to current empirical attacks.  Through this endeavor, we hope our work contributes to the research of privacy-preserving NLP.

\section{Limitations}\label{lm}
The results presented in our paper have two main limitations. First, we have not found a method to calculate the tight ASR bound of the \textit{Contextual Optimal Reconstruction Attack} due to the challenge of accurately computing the theoretical value of $\Pr\left(\mathbf{c}|{x}', y\right)$. Therefore, the calculation of this tight ASR bound remains an open question. Second, the \textit{Contextual Bayesian Attack} entails model fine-tuning and inference. There is still room for a more time-efficient reconstruction attack.

{Furthermore, recent studies on text anonymization~\cite{staab2024large,StaabVBV24}, which leverage advanced LLMs to sanitize sensitive tokens in text, present an exciting avenue for future research. However, the complexity of LLM-based methods introduces unique challenges that require careful consideration. Consequently, the privacy guarantees of these methods are not explored in this paper. We believe that investigating these techniques could provide valuable insights and represent an important direction for future work.}


\section*{Acknowledgments}
This research is supported by the National Natural Science Foundation of China under Grant 62472398, Grant U2336206, and Grant U2436601, as well as by the National Research Foundation, Singapore and Infocomm Media Development Authority under its Trust Tech Funding Initiative (No. DTCRGC-04). Any opinions, findings and conclusions or recommendations expressed in this material are those of the author(s) and do not reflect the views of National Research Foundation, Singapore and Infocomm Media Development Authority.

\bibliography{custom}

\appendix

\label{sec:appendix}
\section{Proof of Theorems}\label{prf}
\begin{theorema}[\textbf{Context-free Optimal Reconstruction Attack}]\label{theorem:1a}
Given a sanitized token $y\in Y$, the optimal strategy to reconstruct the original token of $y$ is to select the token ${x}' \in X$ according to the following rule:
\begin{align}
{x}' =& \underset{{x}' \in X}{\mathrm{arg\,max\,}} \Pr\left({x}'|y\right) \\
=& \underset{{x}' \in X}{\mathrm{arg\,max\,}} \frac{\Pr\left(y|{x}'\right)  \Pr\left({x}'\right)}{\Pr\left(y\right)}.
\end{align}
\end{theorema}

\begin{proof}[Proof of Theorem 1] 
We aim to prove that the attack in \Cref{theorem:1a} maximizes the expected value of the reconstructed token matching the original token for $y$.
Let \(X\) denote the set of probable original tokens for $y$, defined as \(X=\{ x_{1}, x_{2}, \ldots,x_i,\ldots,x_N\}\), where $x_i \in X$ is a probable original token of $y$ with probability $\Pr(x_i|y)$. Let the probability that an attacker infers $x_{i}  \in X$ as the reconstructed token of $y$ be $p_{i}$. It holds that$\sum_{i=1}^N p_{i}=1$. The expected value of the reconstructed token matching the original token of $y$ can be expressed as:
\begin{align}
    &\Pr\left(x_{1}| y\right)p_{1}+\Pr\left(x_{2}| y\right) p_{2}+\ldots+\notag\\
    &\Pr\left(x_{N}| y\right) p_{N}.
\end{align}
With $\sum_{j=1}^N p_{ij}=1$, it can be deduced that:
\begin{align}\label{pr:th1a}
      &\Pr\left(x_{1}| y\right)p_{1}+\Pr\left(x_{2}| y\right) p_{2}+\ldots+\notag\\
    &\Pr\left(x_{N}| y\right) p_{N}
     \leq \underset{{x}' \in X}{\mathrm{\,max\,}} \Pr\left({x}'|y\right).
\end{align}
According to~\Cref{pr:th1a}, the expected value of correctly reconstructing the original token is smaller than $\underset{{x}' \in X}{\mathrm{max,}} \Pr\left({x}'|y\right)$, which is the expected value under attack strategy in~\Cref{theorem:1a}. 
\end{proof}

\begin{theorema}[\textbf{Contextual Optimal Reconstruction Attack}]\label{theorem:2a}
Given a sanitized token $y$ and its sanitized context $\mathbf{c}$, the optimal strategy to reconstruct the original token of $y$ is to select the token ${x}'\in X$ according to the following rule:
\begin{align}
{x}' &= \underset{{x}' \in X}{\mathrm{arg\,max\,}}  \Pr\left({x}'|y,\mathbf{c}\right)\\&= \underset{{x}' \in X}{\mathrm{arg\,max\,}}  \frac{\Pr\left(y|{x}'\right)  \Pr\left({x}'\right)}{\Pr\left(y\right)}  \frac{\Pr\left(\mathbf{c}|{x}', y\right)}{\Pr\left(\mathbf{c}|y\right)}.
\end{align}
\end{theorema}

\begin{proof}[Proof of Theorem 2] 
We aim to prove that the attack in \Cref{theorem:2a} maximizes the expected value of the reconstructed token matching the original token for $y$ and $\mathbf{c}$.
Let \(X\) denote the set of probable original tokens for $y$ and $\mathbf{c}$, defined as \(X=\{ x_{1}, x_{2}, \ldots,x_i,\ldots,x_N\}\), where $x_i\in X$ is a probable original token of $y$ and $\mathbf{c}$ with probability $\Pr(x_i|y,\mathbf{c})$. Let the probability that an attacker inferring $x_{i} \in X$ as the reconstructed token of $y$ and $\mathbf{c}$ be $p_{i}$ and it holds that$\sum_{i=1}^N p_{i}=1$. The expected value of the reconstructed token matching the original token of the sanitized token $y$ and  sanitized context $\mathbf{c}$ can be expressed as:
\begin{align}
    &\Pr(x_{1}| y,\mathbf{c}) p_{1}\!+\!\Pr(x_{2}| y,\mathbf{c}) p_{2}\!+\!\ldots\!+\notag\\&\Pr(x_{N}| y,\mathbf{c}) p_{N}.
\end{align}
Utilizing $\sum_{i=1}^N p_{i}=1$, it can be deduced that:
\begin{align}\label{pr:th2a}
     \sum_{i=1}^{N} \Pr(x_i \mid y, \mathbf{c})p_{i}\leq \underset{{x}' \in X}{\mathrm{\,max\,}} \Pr\left({x}'|y,\mathbf{c}\right).
\end{align}
According to~\Cref{pr:th2a}, the expected value of correctly inferring the original token is smaller than $\underset{{x}' \in X}{\mathrm{max,}} \Pr\left({x}'|y,\mathbf{c}\right)$, which is the expected value under attack strategy in~\Cref{theorem:2}. 
\end{proof}

\section{Discussion of \textbf{Contextual \textit{K}-ASR Bound}}\label{tmk}
\begin{table*}[ht]
\centering
\caption{Time cost of reconstructing $1,000$ tokens with various $K$ in the \textit{Contextual \textit{K}-ASR Bound}.}
\label{tbl:time2}
\resizebox{0.54\textwidth}{!}{
\begin{tabular}{c|ccccc}
\toprule
\textbf{\textit{K}} & $K=1$ & $K=5$ & $K=10$ & $K=15$ & $K=20$ \\ 
\midrule
\textbf{Minute} & $11.631$ & $16.659$ & $18.850$ & $24.566$ & $29.216$ \\
\bottomrule
\end{tabular}
}
\end{table*}

We assess the time required to calculate the Contextual \textit{K}-ASR Bound for $1,000$ sanitized tokens. we conduct experiments using \textit{CUSTEXT+}~\cite{chen2023customized} defense on SST-2~\cite{socher-etal-2013-recursive} dataset with various values of $K$. We set the size of $X$ in \textit{CUSTEXT+} to its default $20$. The experimental results are shown in \Cref{tbl:time2}. It is observed that as $K$ increases, the time cost also increases.

On the other hand, the size of the input token set $X$ in  \textit{SANTEXT+}~\cite{yue-etal-2021-differential} far exceeds 20. Specifically, it uses an input token set $X$ consisting of more than $15,000$ tokens in SST-2.  Given this, traversing the entire $X$ for the Contextual \textit{K}-ASR Bound in \textit{SANTEXT+} is not practical to implement.

\section{Discussion of {Context-free Bayesian Attack}}\label{dsa}

The attack strategy in \Cref{theorem:1} reconstructs a token considering two factors: a larger $\Pr(y|{x}')$ and a larger $\Pr({x}')$. However, when the probabilities $\Pr({x}')$ of all probable original tokens corresponding to $y$ are zero in the shadow dataset, the attack strategy in \Cref{theorem:1} loses its effectiveness of reconstructing a token with larger $\Pr(y|{x}')$. In such instances,  \textit{Context-free Bayesian Attack}, incorporating $\alpha^{-1}$, retains its effectiveness. This method selects a reconstructed token ${x}'$ according to:
\begin{align}
{x}' &=\underset{{x}' \in X}{\mathrm{arg\,max}\,} \frac{\Pr(y|{x}') \cdot(0+\alpha^{-1})}{\Pr(y)} \\&=\underset{{x}' \in X}{\mathrm{arg\,max}\,} \frac{\Pr(y|{x}') \cdot\alpha^{-1}}{\Pr(y)}.
\end{align}
Furthermore, $\alpha^{-1}$ has minimal impact on the outcome of the reconstruction attack when these probabilities $\Pr({x}')$ are not zero.

To evaluate the practical impact of the parameter $\alpha^{-1}$, we conducted an ablation study across three datasets: SST-2~\cite{socher2013recursive}, QNLI~\cite{DBLP:conf/iclr/WangSMHLB19}, and AGNEWS~\cite{zhang2015character}. We assessed the Attack Success Rate (ASR) of the \textit{Context-free Bayesian Attack} both with and without $\alpha^{-1}$ against \textit{SANTEXT+}~\cite{yue-etal-2021-differential}, creating $3,000$ test samples randomly selected from each dataset.

Figure~\ref{DS} presents the results of the ablation studies on the parameter $\alpha^{-1}$, where the experimental results demonstrate the effectiveness of $\alpha^{-1}$.

\begin{table*}[ht]
\centering
\caption{ASR of \textit{Context-free Bayesian Attack} using various constants instead of $\alpha^{-1}$ against \textit{SANTEXT+} defense.}
\label{tbl:diff_size}
\resizebox{\textwidth}{!}{
\begin{tabular}{c||c c c c|c c c c|c c c c}
\toprule
\multirow{2}{*}{\textbf{Constant}} & \multicolumn{4}{c|}{\textbf{SST-2}} & \multicolumn{4}{c|}{\textbf{QNLI}} & \multicolumn{4}{c}{\textbf{AGNEWS}} \\
\cline{2-13}
   \rule{0pt}{2.3ex}
& $\epsilon=2.0$ & $\epsilon=4.0$ & $\epsilon=2.8$ & $\epsilon=3.2$ & $\epsilon=2.0$ & $\epsilon=4.0$ & $\epsilon=2.8$ & $\epsilon=3.2$ & $\epsilon=2.0$ & $\epsilon=4.0$ & $\epsilon=2.8$ & $\epsilon=3.2$ \\ 
\midrule
\textbf{\large{$\alpha^{-1}$}} &$0.180$ &$0.410$&$0.650$  &$0.843$ &$0.155$ &$0.434$ &$0.711$ & $0.865$ &$0.223$ & $0.450$ &$0.712$ &$0.894$\\
{\large{w/o $\alpha^{-1}$}} &$0.130$ &$0.331$ &$0.600$ &$0.682$  &$0.133$ &$0.387$ &$0.617$ & $0.780$ &$0.204$ &$0.416$ &$0.653$ &$0.820$ \\
\textbf{\large{$0.1\times\alpha^{-1}$}} &$0.173$ &$0.404$  &$0.648$ &$0.840$ &$0.152$&$0.429$ &$0.704$ &$0.862$ & $0.218$ &$0.448$ &$0.709$ & $0.892$\\
\textbf{\large{$5\times\alpha^{-1}$}} &$0.180$ &$0.410$ &$0.650$ &$0.843$ &$0.155$ &$0.434$ & $0.711$ & $0.865$ & $0.222$ &$0.449$ &$0.711$ & $0.894$\\
\textbf{\large{$10\times\alpha^{-1}$}} &$0.180$ &$0.410$&$0.650$&$0.843$ &$0.155$ &$0.434$ & $0.711$ & $0.865$  &$0.221$& $0.449$ & $0.711$&$0.894$\\
\bottomrule
\end{tabular}
}
\label{DS_a}
\end{table*}
To further explore the influence of different constants instead of $\alpha^{-1}$ on \textit{Context-free Bayesian Attack}, we compare the ASR of $\alpha^{-1}$ with those of other constants against \textit{SANTEXT+}. 

\Cref{DS_a} illustrates the experimental results across three datasets under various privacy parameters \(\epsilon\). These results highlight the effectiveness of the constant introduced in the \textit{Context-free Bayesian Attack}. When this constant is smaller than \(\alpha^{-1}\), the ASR of the \textit{Context-free Bayesian Attack} decreases slightly in the experiments. Conversely, when this constant exceeds \(\alpha^{-1}\), the ASR does not show significant growth.

\section{Discussion of {Contextual Bayesian Attack}}\label{clc}

To evaluate the effectiveness of the detector $h$, we calculate the classification accuracy of both the detector $h$ and the original BERT model \cite{devlin2018bert} across three datasets: SST-2~\cite{socher2013recursive}, QNLI~\cite{DBLP:conf/iclr/WangSMHLB19}, and AGNEWS~\cite{zhang2015character}. We construct $3,000$ test samples that are disjointed from the training data of detector $h$ and randomly selected. For the construction of training samples,  we generate 100 sanitized sentences for each original sentence in the shadow dataset. The experiment is conducted $10$ times to calculate standard deviations (1-sigma) using the closed form formula \cite{altman2005standard} with numpy\footnote{\url{https://numpy.org/}}.

\Cref{fig:2} shows the experimental results. It is observed that the classification accuracy of $h$ is higher than that of the original BERT model. Additionally, the classification accuracy of $h$ increases with the privacy parameter $\epsilon$. These results demonstrate the effectiveness of the detector $h$ in inferring whether $\mathbf{c}$ is the sanitized context of ${x}'$ and $y$.

 \begin{figure*}[ht]
\hspace{0.1cm}
\includegraphics[width=0.95\textwidth]{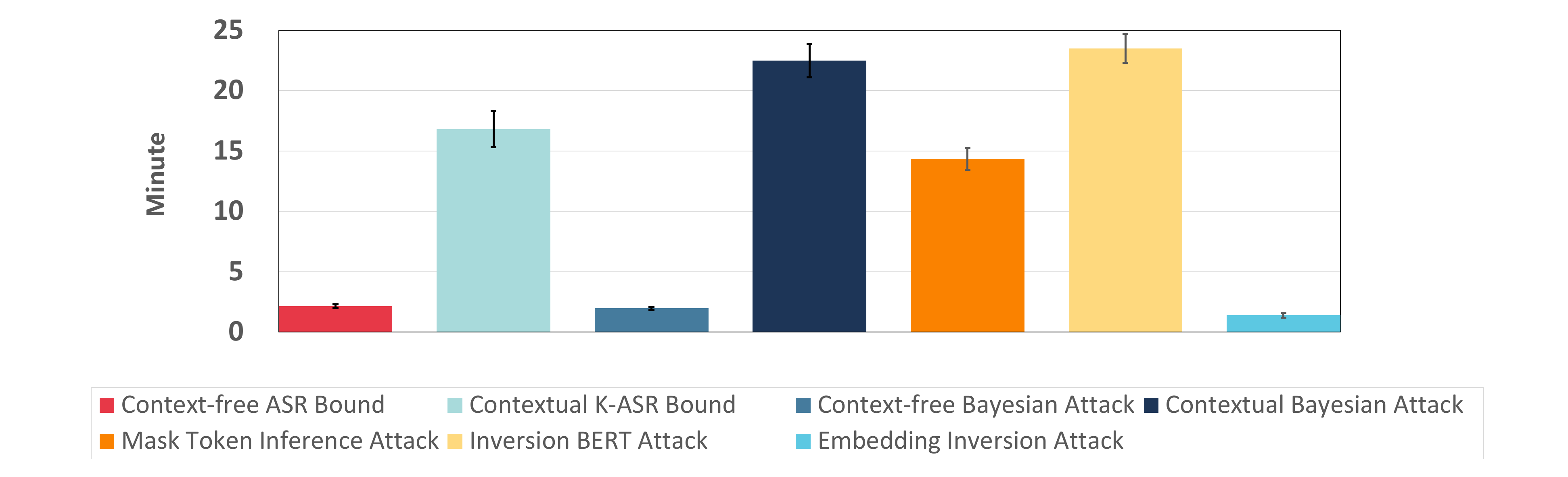}
\caption{Comparison of the time costs in various reconstruction attacks.}
\label{fig:c2}
\end{figure*}
\section{Experimental dataset}\label{dats}
We conduct the experiments of reconstruction attacks on the following four datasets:
\begin{itemize}
\item \textbf{SST-2~\cite{socher2013recursive} :} The (binary) Stanford Sentiment Treebank (SST-2) dataset consists of movie reviews with sentiment labels. This dataset is used to categorize the sentiment of sentences~\cite{DBLP:conf/acl/TianGXLHWWW20} into the positive or negative label. We use the version in the GLUE~\cite{DBLP:conf/iclr/WangSMHLB19}.

\item \textbf{QNLI~\cite{DBLP:conf/iclr/WangSMHLB19} :} The Question-answering Natural Language Inference (QNLI) dataset is derived from the Stanford Question Answering Dataset (SQuAD)~\cite{DBLP:conf/emnlp/RajpurkarZLL16}. This dataset consists of question-paragraph pairs, where the task is to determine whether the paragraph contains the answer to the question. We also utilize the use available in the GLUE benchmark~\cite{DBLP:conf/iclr/WangSMHLB19}.

\item \textbf{AGNEWS~\cite{zhang2015character} :} The AG News dataset is a widely used benchmark for text classification tasks. It consists of news articles and their corresponding labels. This dataset is utilized to train a model that classifies news into four categories: World, Sports, Business, and Sci/\,Tech. We use the version of this dataset available on this site\footnote{\href{http://groups.di.unipi.it/~gulli/AG_corpus_of_news_articles.html}{AG Corpus of News Articles}}. 

\item \textbf{Yelp~\cite{zhang2015character} :} The Yelp dataset consists of reviews from the Yelp website, where the task is to predict the sentiment expressed in the reviews. The dataset is used for sentiment analysis tasks. We utilized the version available from the Yelp Challenge\footnote{\href{https://www.yelp.com/dataset}{Yelp Open Dataset}}. 
\end{itemize}

\section{Implementation details}\label{imp}
\subsection{Context-free ASR Bound and Contextual \textit{K}-ASR Bound}
 We run experiments on a cluster with NVIDIA RTX A6000 GPUs and  Intel Xeon Gold 6130 2.10 GHz CPUs. We conduct the implementation of all mechanisms using Python. We calculate $\Pr({x}')$ by estimating the statistical probability of the original sentences corresponding to the sanitized tokens.

 In the \textit{Contextual K-ASR Bound}, we generate 30 sanitized sentences for each original sentence to construct training samples. For each token ${y} \in \Gamma$ in the sanitized sentences, we create a training sample. The input for each sample consists of two token sequences: The first sequence is the sanitized sentence consisting of ${y}$ and $\mathbf{c}$. The second sequence replaces the token ${y} \in \Gamma$ in the first sequence with its reconstructed result ${x}'$ from \textit{Context-free Bayesian Attack}. If ${x}'$ matches the original token of ${y}$, which means that $\mathbf{c}$ is the sanitized context of the sanitized token $y$ and the reconstructed token ${x_i}'$, we label the sample $1$. Otherwise, we label it $0$. To maintain label balance, for each sample labeled as $1$, we generate a sample labeled as $0$ using the same sanitized sentence with the second most likely token from \textit{Context-free Bayesian Attack}. We also generate a sample labeled as $1$ with the original token for each sample labeled as $0$. For model fine-tuning, we use the pooled output from the BERT model~\cite{devlin2018bert} and transform it into class scores through a linear layer. 
 
We use the \textit{Adam\,Optimizer} \cite{kingma2014adam} with a learning rate of 0.00005. We reduce the learning rate by 0.1 every 1000 steps. We measure performance with \textit{CrossEntropyLoss}~\cite{zhang2018generalized} and fine-tune the model over 3 epochs on constructed training samples. The fine-tuned model represents probability $\Pr(\mathbf{c} | x', y)$ that $\mathbf{c}$ is the sanitized context for $x'$ and $y$, based on its output of the label $1$.

\subsection{Context-free Bayesian Attack and Contextual Bayesian Attack}
We run experiments on a cluster with NVIDIA RTX A6000 GPUs and  Intel Xeon Gold 6130 2.10 GHz CPUs. We conduct the implementation of all mechanisms using Python. We calculate the statistical probability in the shadow dataset to represent the approximated probability $\Pr({x}')$. 

In the \textit{Contextual Bayesian Attack}, we generate 100 sanitized sentences for each original sentence to construct training samples. For each token ${y} \in \Gamma$ in the sanitized sentences, we create a training sample. The input for each sample consists of two token sequences: The first sequence is the sanitized sentence consisting of ${y}$ and $\mathbf{c}$. The second sequence replaces the token ${y} \in \Gamma$ in the first sequence with its reconstructed result ${x}'$ from \textit{Context-free Bayesian Attack}. If ${x}'$ matches the original token of ${y}$, which means that $\mathbf{c}$ is the sanitized context of the sanitized token $y$ and the reconstructed token ${x}'$, we label the sample $1$. Otherwise, we label it $0$. To maintain label balance, for each sample labeled as $1$, we generate a sample labeled as $0$ using the same sanitized sentence with the second most likely token from \textit{Context-free Bayesian Attack}. We also generate a sample labeled as $1$ with the original token for each sample labeled as $0$. For model fine-tuning, we use the pooled output from the BERT model~\cite{devlin2018bert} and transform it into class scores through a linear layer. 

We use the \textit{Adam\,Optimizer} \cite{kingma2014adam} with a learning rate of 0.00005. We reduce the learning rate by 0.1 every 1000 steps. We measure performance with \textit{CrossEntropyLoss}~\cite{zhang2018generalized} and fine-tune the model over 3 epochs on constructed training samples. The fine-tuned model represents probability $\Pr(\mathbf{c} | x', y)$ that $\mathbf{c}$ is the sanitized context for $x'$ and $y$, based on its output of label $1$.

\section{Comparison of time cost}\label{ct22}

To further compare the computation costs, we estimated the time required for various reconstruction attacks to reconstruct $1,000$ sanitized tokens on the SST-2 dataset using the \textit{CUSTEXT+} as a defense. Similarly, we conducted this experiment $10$ times to calculate the standard deviations (1-sigma) using the closed form formula~\cite{altman2005standard} with the numpy\footnote{\url{https://numpy.org/}}.

As illustrated in \Cref{fig:c2}, the time cost for both the \textit{Embedding Inversion Attack} and the \textit{Context-free Bayesian Attack} is approximately comparable, both being less than 5 minutes. These two reconstruction attacks do not require model fine-tuning or inference. In contrast, the \textit{Contextual Bayesian Attack}, \textit{Mask Token Inference Attack}, and \textit{Inversion BERT Attack} involve either model fine-tuning or model inference. Among these three, the \textit{Mask Token Inference Attack} has the smallest time cost, roughly 15 minutes; the \textit{Inversion BERT Attack} takes the longest, about 25 minutes. The time cost of the \textit{Contextual Bayesian Attack} in the experiments is greater than that of the \textit{Mask Token Inference Attack}. {This is because \textit{Mask Token Inference Attack} only involves an inference phase without model training. However, our contextual attack requires training a model and inferring on this model. Therefore, \textit{Contextual Bayesian Attack} is slower than the \textit{Mask Token Inference Attack}.} Moreover, it is observed that the standard deviations of time costs for the \textit{Embedding Inversion Attack} and the \textit{Context-free Bayesian Attack} are smaller than those for the other reconstruction attacks.

\section{Reconstruction attacks with the large epsilon}\label{ct24}
{As illustrated in \Cref{tab:attack-result}, the performance of our proposed attacks is comparable to that of the \textit{Embedding Inversion Attack} for large values of $\epsilon$. The rationale behind this phenomenon is as follows:}

{Let $x$ represent a raw token and $y$ its corresponding sanitized token.  The performance of our attacks partially depends on $\Pr(y|x)$. Specifically, as $\epsilon$  increases, $\Pr(y|x)$ also increases in both \textit{SANTEXT+} and \textit{CUSTEXT+}, causing our attacks to tend to select the token $x$  with the highest $\Pr(y|x)$  as the reconstructed token. For a given sanitized token $y$, the \textit{Embedding Inversion Attack} selects the token $x'$  with the smallest embedding distance to $y$  as the reconstructed token. In the context of \textit{SANTEXT+} and \textit{CUSTEXT+}, this is equivalent to selecting the token $x'$  with the highest $\Pr(y|x')$.}

{From the above analysis, it can be concluded that when $\epsilon$ is large, our proposed attack and the \textit{Embedding Inversion Attack} tend to select the same token as the reconstructed token in most cases.}

\section{Reconstruction attacks with the misaligned 
shadow dataset}\label{ct23}

\begin{figure}[ht]
\hspace{0cm}
\includegraphics[width=0.95\columnwidth]{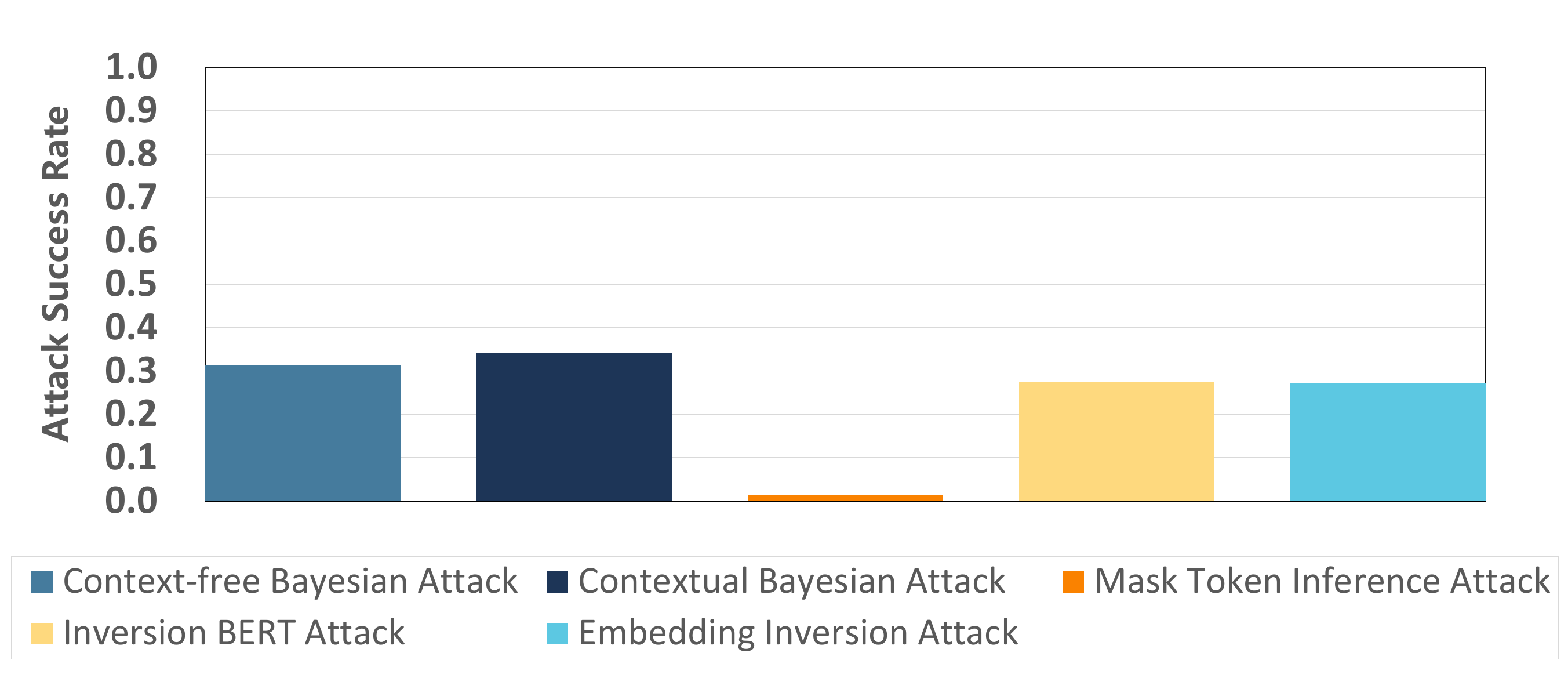}
\caption{ASR of reconstruction attacks with the misaligned 
shadow dataset.}
\label{fig:c3}
\end{figure}

{In this paper, we present a privacy evaluation study on DP-based text sanitization. Our threat model assumes that practical attacks access a shadow dataset from the same distribution as the target to reconstruct. This assumption simulates the worst-case scenario for privacy preservation and is consistent with the settings described in previous work~\cite{balle2022reconstructing, shokri2017membership}.}

{To reflect real-world scenarios where shadow datasets may not be perfectly aligned with sanitized data, we conduct additional attacks on \textit{CUSTEXT+} with $\epsilon = 3.0$ (as recommended by its authors) using a misaligned shadow dataset. Specifically, we use SST-2 as the shadow dataset and QNLI as the target dataset to reconstruct. \Cref{fig:c3} presents the experimental results with the misaligned shadow dataset. It is observed that our reconstruction attacks remain comparable to the baseline methods.}
\end{document}